\documentclass[12pt]{extarticle}
\usepackage{amsmath, amsthm, amssymb,mathtools, color,bm}
\usepackage[dvipsnames,table]{xcolor}
\usepackage{chemfig}
\usepackage[ pdfborderstyle={/S/U/W 1}]{hyperref}
\definecolor{MyBlue}{HTML}{210cac}
\definecolor{MyCommentColor}{HTML}{E82929}
\definecolor{MyCiteColor}{HTML}{0099FF}
\definecolor{MyRed}{HTML}{3E186A}
\hypersetup{%
  pdfborderstyle={/S/U/W 1},
  colorlinks=true,
  linkcolor=MyBlue,
  citecolor=MyCiteColor,
  urlcolor=MyRed,
  anchorcolor=MyBlue,
  linkbordercolor=MyRed,
}
\usepackage[shortlabels]{enumitem}
\usepackage{graphicx}
\usepackage{makecell}
\usepackage[final]{pdfpages}
\setboolean{@twoside}{false}
\usepackage{pdfpages}
\usepackage{tikz}

\usetikzlibrary{shapes,arrows,spy,positioning,decorations}
\usepackage{caption}
\usepackage{subfig}
\usepackage{mathrsfs}
\usepackage{scalefnt}
\usepackage{verbatim}
\oddsidemargin \evensidemargin
\makeatletter
\newcommand{\eqnum}{\refstepcounter{equation}\textup{\tagform@{\theequation}}}
\makeatother

\newtheorem{theorem}{Theorem}
\numberwithin{theorem}{section}
\newtheorem{proposition}[theorem]{Proposition}
\newtheorem{lemma}[theorem]{Lemma}
\newtheorem{corollary}[theorem]{Corollary}
\newtheorem{definition}[theorem]{Definition}

\newtheorem{remark}[theorem]{Remark}
\newtheorem{example}[theorem]{Example}

\newcommand{\RR}{\mathbb{R}}

\newcommand{\CC}{\mathbb{C}}
\newcommand{\ZZ}{\mathbb{Z}}

 \date{}

\begin{document}

\title{
Families of Toric Chemical Reaction Networks}
\author{Michael F.~Adamer and Martin Helmer}
\maketitle
\begin{abstract}
We study families of chemical reaction networks whose positive steady states are toric, and therefore can be parameterized by monomials. Families are constructed algorithmically from a core network; we show that if a family member is multistationary, then so are all subsequent networks in the family. Further, we address the questions of model selection and experimental design for families by investigating the algebraic dependencies of the chemical concentrations using matroids. 
Given a family with toric steady states and a constant number of conservation relations, we construct a matroid that encodes important information regarding the steady state behaviour of the entire family. Among other things, this gives necessary conditions for the distinguishability of families of reaction networks with respect to a data set of measured chemical concentrations. We illustrate our results using multi-site phosphorylation networks.
\end{abstract}

\section{Introduction}
Many of the fundamental processes in biological cells can be described by a set of interlinked chemical reactions. Prominent examples of cellular processes regulated via biochemical interactions include immune response \cite{McKeithan1995}, cell signalling \cite{Salazar2009}, cell death \cite{Ferrell1997,Burack1997}, and toxin formation \cite{Lee2016}. For this reason the study of chemical reaction networks forms a central part of algebraic systems biology \cite{Gross2016,Gunawardena2003,Gunawardena2009,Holstein,Yu2018}. One approach focuses on the long term behaviour of networks by investigating their steady states and the relation of the number and stability of steady states to the network structure \cite{Yu2018,Craciun2009,ProcPhos}. In this paper we investigate the positive steady states for algorithmically constructed reaction networks, which we call \textit{families}, for which the positive steady states may be parameterized by monomials (i.e.~polynomials with a single term). Further, we use the algebraic dependencies of the variables representing the chemical concentrations to investigate experimental design and model identification for entire families.

Families of networks are formally defined in Definition \ref{def:recursiveG}. To obtain an intuition for what could be described as a family, we introduce phosphorylation networks \cite{ProcPhos,PerezMillan}. Phosphorylation is a vital signalling process in biochemistry and it is one of the most widely studied protein modifications. During phosphorylation a phosphoryl group ($PO_3^-$) is added to an organic molecule which acts as a substrate. The chemical reaction is catalysed by enzymes and often many phosphoryl groups can be added to the same substrate. This process is known as multisite phosphorylation. There are two limiting mechanisms which have been investigated in the literature and which will serve as the running examples in this paper, namely, distributive phosphorylation and processive phosphorylation. In the distributive system the enzyme unbinds from the substrate after every time a $PO_3^-$ is added and in the processive mechanism the enzyme stays bound until the substrate is fully phosphorylated. The desphosphorylation mechanisms work analogously. The distributive and processive mechanisms for one site substrates follow the same reaction scheme
\begin{align}
\begin{split}
  &E+S_0\rightleftharpoons ES_0 \rightarrow E+S_1,\\
  &F+S_1\rightleftharpoons FS_1 \rightarrow F+S_0.
  \label{netw:Phos1}
  \end{split}
\end{align}
However, on a two-site substrate the mechanisms differ, namely, the distributive mechanism is given by
\begin{align}
  \begin{split}
  E+S_0\rightleftharpoons ES_0 \rightarrow E+S_1 &{\color{OliveGreen}\rightleftharpoons ES_1 \rightarrow E+S_2},\\
  &{\color{OliveGreen}F+S_2\rightleftharpoons FS_2 \rightarrow} F+S_1\rightleftharpoons FS_1\rightarrow F+S_0,
  \label{netw:DistPhos2}
  \end{split}
\end{align}
and the processive system is
\begin{align}
    \begin{split}
        E+S_0\rightleftharpoons ES_0 &{\color{OliveGreen}\rightleftharpoons ES_1 \rightarrow} E+S_2,\\
        F+S_2 &{\color{OliveGreen}\rightleftharpoons FS_2} \rightarrow FS_1 \rightarrow F+S_0.
    \label{netw:ProcPhos2}
    \end{split}
\end{align}
It is clear that the constructions for both the distributive and the processive system can be extended to an $N$-site substrate and that two different procedures are needed to do so. Hence, we can consider these two mechanisms to be in two distinct families of networks. In this paper we study graph theoretic constructions which allows us to rigorously identify families of networks and also construct them from a base graph. Most importantly, though, we investigate which steady state properties are conserved throughout a family.

One of the central goals of this paper is to develop criteria to establish which families have members with multiple positive steady states, so-called multistationary networks. Confirming that a network is multistationary and finding the associated parameter regions is highly nontrivial; a range of different approaches have been applied previously (see \cite{Joshi2015} for a survey). 
In particular, monomial positive steady state parameterizations have proved fruitful due to their relations to toric varieties which are well understood in algebraic geometry \cite{FultonToric}.

Previous work relating to the concept of families in this paper considers so-called atoms of multistationarity, which are the smallest multistationary subnetworks which can induce multistationarity in their parent networks \cite{Joshi2013}. Network properties resulting from the gluing of networks are investigated in \cite{Gross2018}. Other network modifications which preserve or destroy multistationarity are studied in \cite{Banaji2018}. Recent results extend the techniques for identifying multistationarity to highly structured networks \cite{Dickenstein2018} and networks with intermediate species \cite{Sadeghimanesh2018,Feliu2013}.
In this paper we develop a concept of families of networks which unifies the notions of highly structured networks, embedded subnetworks as defined in \cite{Joshi2013} (i.e.~distributive systems), and networks with intermediate species \cite{Feliu2013} (i.e.~processive systems). We show in Theorem \ref{thm:MultiStat} that if a member of a family is multistationary then so are all larger networks obtained by the recursive construction of Definition \ref{def:recursiveG}.


Going beyond multistationarity, we also investigate the necessary conditions for model rejection among members of a family if only limited steady state data is available. In particular, in Section \ref{sec:matroids}, we encode algebraic dependencies between the variables using a combinatorial object called an algebraic matroid \cite{Oxley,Rosen}. From the algebraic matroid we find binomial relations which have to be satisfied by the chemical concentrations at any positive steady state, so-called steady state invariants. The results in Section \ref{Sec:ExpermentialDesign} extend the previous research of \cite{Harrington2012,MacLean2015,Karp2012} and the application of matroids for experimental design presented in \cite{Gross2016,MacLean2015}. Using these previous results, we give novel necessary conditions for the distinguishability of two members of a family of reaction networks with respect to a data set of measured chemical concentrations. Consequently, we can also prescribe which species to measure to be able to reject a family member.

In summary, the biochemical questions we would like to address are:
\begin{enumerate}
\item What conditions are needed to define families of chemical reaction networks and what is the relation between their steady states? (Section \ref{sec:FamiliesOfGraphs})
\item Can we use the family construction in model selection or parameter estimation for the entire family? (Section \ref{sec:matroids})
\item Can we find conditions such that multistationarity of one family member implies multistationarity for all subsequent members? (Section \ref{sec:Multi})
\end{enumerate}
These motivating questions from chemistry translate into the following algebraic questions which we answer using techniques from toric geometry and matroid theory:
\begin{enumerate}
\item What are the relations between the toric varieties defined by recursively constructed reaction graphs?
\item What is the connection between the circuit polynomials of matroids associated to different family members?
\item What is the relation between the positive parts of the steady state varieties of subsequent family members?
\end{enumerate}

This paper is organised as follows. Section \ref{sec:background} introduces chemical reaction networks and relevant definitions from toric geometry and matroid theory. In Section \ref{sec:FamiliesOfGraphs} we give a rigorous definition of a family of reaction network graphs and some preliminary results. In Section \ref{sec:matroids} we focus on biochemical and algebraic question 2 using matroid theory. We also introduce some new terminology which simplifies the proofs in the remainder of the paper. In Section \ref{sec:Multi} we prove the main result on multistationarity using the matroidal language developed in the previous section. We discuss our results and suggest further directions in Section \ref{sec:Conclusion}.

\section{Background}\label{sec:background}

In this section we briefly introduce aspects of chemical reaction network theory and discuss essential notions of algebraic geometry and matroid theory. 

\subsection{Chemical Reaction Network Theory}\label{sec:CRNT}

Informally a chemical reaction network (CRN) can be described by a multiset $\mathfrak{N} = \{\mathcal{S},\mathcal{C},\mathcal{R}\}$, where $\mathcal{S}$ is the set of species, $\mathcal{C}$ is the set of linear combinations of species (complexes) and $\mathcal{R}$ is the set of reactions.


\begin{example}[Michaelis-Menten]
The set $\mathcal{S}$ of chemical species present in the network 
$$E+S \rightleftharpoons ES \rightarrow E+P$$
is defined by $\mathcal{S} = \{E,S,ES,P\}$. The complexes, which are linear combinations of species, are $\mathcal{C} = \{E+S,ES,E+P\}$. The reaction set is $\mathcal{R} = \{E+S\to ES,ES\to E+S,ES\to E+P\}$.
\end{example}
The multiset $\mathfrak{N}$ defines a directed graph (digraph) $\mathcal{G}$ whose vertex set is $\mathcal{C}$ and whose edge set is defined by the reaction set $\mathcal{R}$. The reaction from complex $C_i$ to $C_j$  is an element of the reaction set $\mathcal{R}$ if and only if there is a directed edge $C_i\to C_j$ in $\mathcal{G}$. Let $X_l\in\mathcal{S}$ and $\{\alpha_{il}\}\in \ZZ_{\geq 0}$. A reaction from complex $C_i = \sum_l \alpha_{il}X_l$ to $C_j = \sum_l\alpha_{jl}X_n$, with rate constant $\kappa$ is written as
\begin{equation}
  \sum_l \alpha_{il}X_l\xrightarrow{\kappa} \sum_l \alpha_{jl}X_l;
  \label{eq:GenericReac}
\end{equation}the constants $\alpha_{il}$ are called the stoichiometric coefficients of the complex $C_i$.
 Let the reaction vector for the $\ell^\text{th}$ reaction $C_i\to C_j$ be $r_\ell = \alpha_j-\alpha_i$ where $\alpha_i$, $\alpha_j$ are the column vectors of the stoichiometric coefficients of the complexes $C_i$ and $C_j$. The $n\times m$ matrix of all reaction vectors $\Gamma = (r_1,\dots, r_m)$ is called the {\em stoichiometric matrix}. The rate constant $\kappa$ assigns a weight to each edge of the digraph $\mathcal{G}$, making $\mathcal{G}$ a weighted digraph.
Definition \ref{def:ReacGraph} gives the description of a CRN which we are going to adopt for the remainder of this paper. While this is not the standard definition of a CRN, but equivalent to it, it allows us to focus more on the graphical structure of CRNs.

\begin{definition}
  A chemical reaction network is a weighted directed graph $\mathcal{G} = (\mathcal{C},\mathcal{R})$ with vertex set $\mathcal{C}$, edge set $\mathcal{R}$ and edge weights $\kappa=(\kappa_1,\dots,\kappa_m)^T$.
  \label{def:ReacGraph}
\end{definition}

To connect the graphical structure of a CRN to its dynamical properties a description of reaction kinetics is needed. Previous work introduced a number of reaction laws such as mass action, rational function kinetics, Michaelis-Menten kinetics or Hill function kinetics \cite{Cornish}. In this paper we use mass action kinetics
which assigns a monomial to each complex in the network. Let the chemical concentration of chemical species $X_n$ be $x_n$, then the monomial for complex $C_i$ is obtained by
\begin{equation}
  x^{\alpha_i} = x_1^{\alpha_{i1}}\cdots x_n^{\alpha_{in}}.
\end{equation}
Hence, using the representation of complexes as monomials we represent $\mathcal{C}$ as a vector of monomials $x^\alpha = (x^{\alpha_1},\dots,x^{\alpha_m})^T$.

\begin{example}\label{Ex:Phos_1}
  Revisiting the one site phosphorylation network \eqref{netw:Phos1}, we map each species to its concentration $E\to x_1,\:F\to x_2,\: S_0\to x_3,\: S_1\to x_4,\: ES_0\to x_5,\: FS_1\to x_6$ and introduce a vector of rate constants $\kappa = (\kappa_1,\dots,\kappa_6)^T$. Then, the reaction network is represented by the weighted digraph
  \begin{align*}
  &x_1x_3\xrightleftharpoons[\kappa_2]{\kappa_1} x_5 \xrightarrow{\kappa_3} x_1x_4,\\
  &x_2x_4\xrightleftharpoons[\kappa_5]{\kappa_4} x_6 \xrightarrow{\kappa_6} x_2x_3.
  \end{align*}
\end{example}

The dynamics of the network can be expressed in terms of the network structure and the stoichiometric coefficients as a set of ordinary differential equations
\begin{equation}
  \frac{dx}{dt} = \alpha^T A_\kappa x^\alpha
  \label{eq:DynSys}
\end{equation}
where $A_\kappa$ is the negative weighted graph Laplacian of $\mathcal{G}$, $\alpha^T$ is the matrix of stoichiometric coefficients and $x^\alpha$ is a vector of monomials. An alternative representation of equation \eqref{eq:DynSys} assigns a monomial $\kappa_\ell x^{\alpha_\ell}$ to the $\ell^\text{th}$ reaction in the network to build a vector $R(x) = (\kappa_1x^{\alpha_1},\dots,\kappa_mx^{\alpha_m})^T$. Then, the dynamical system \eqref{eq:DynSys} is given by $dx/dt = \Gamma R(x)$, where $\Gamma$ is the stoichiometric matrix as above.

The left kernel of the stoichiometric matrix, $\Gamma$, is of biochemical importance as it describes conservation relations. Conservation relations induce linear relations between the variables. Informally we say that a set of species is conserved if their {total concentration, $c$, is constant}. In particular, suppose $z\in \ker\left(\Gamma^T\right)$, then $z^T(dx/dt) =0$ which implies $z^T x = c \in \RR_{>0}$; this latter equation is then a conservation relation.

\begin{remark}
Given a CRN, $\mathfrak{N}$, with a maximum of $d$ independent conservation relations, the linear subspace defined by the conservation relations, often referred to as compatibility class, may be compactly written as $Zx-c=0$, where $c\in\RR_{\geq 0}^d$. The rows of $Z x-c$ define the subspace and the matrix $Z = (z_1,\dots,z_d)^T$ represents the conservation relations.
\label{rem:ZMatrix}
\end{remark}

\begin{example}
In the one site phosphorylation system \eqref{netw:Phos1} (also Example \ref{Ex:Phos_1}), the enzyme $E$ is conserved, but can exist in two states, the free state $E$ and the bound state $ES_1$. Thus, the total concentration $c_1 = x_1 + x_3$ is conserved. In total there are three independent conservation relations as the total mass of $E$, $F$, and the substrate is conserved respectively.
\end{example}

We now proceed to defining the steady states of a chemical reaction network \cite{PerezMillan}.

\begin{definition}
  A vector  $x^*\in \CC^n$ is called a steady state of a CRN if $\alpha^T A_\kappa \left(x^*\right)^\alpha = 0$. A positive steady state is a steady state such that $x^*\in \RR_{>0}^n$ and $\alpha^T A_\kappa \left(x^*\right)^\alpha = 0$.
\end{definition}

Often one is interested in whether a chemical reaction network can have multiple positive steady states for a given set of rate constants $\kappa$ and total concentrations $c$. 

\begin{definition}
A chemical reaction network is multistationary if there exists a set of parameters $\{\kappa_1,\dots,\kappa_m\}$ such that $\alpha^T A_\kappa \left(x^*\right)^\alpha = \alpha^T A_\kappa \left(y^*\right)^\alpha = 0$ and $x^*-y^* \in \ker (Z)$ for two distinct positive steady states $x^*$ and $y^*$.
\end{definition}







\subsection{Algebraic Geometry and Toric Steady States}\label{sec:Toric}

In this subsection we briefly introduce the notion of toric varieties from algebraic geometry, and their connections to chemical reaction networks. For an introduction to toric varieties we refer the reader to \cite{FultonToric,CLS_Toric}. 

Given polynomial equations $f_1(x_1,\dots, x_n),\dots, f_r(x_1,\dots, x_n)$ in the polynomial ring $\CC[x_1,\dots, x_n]$ we define the {\em algebraic variety} 
\begin{equation}
    V(f_1,\dots, f_r)=\{ x\in \CC^n\;|\; f_1(x)=\cdots =f_r(x)=0 \}.
\end{equation}
Note that, by definition, $g(x)=0$ for any $x\in V(f_1,\dots, f_r)$ and any polynomial $g$ in the {\em ideal} $\langle f_1,\dots, f_r\rangle=\{\sum_i h_i f_i \;|\: h_i \in\CC[x_1,\dots, x_n] \}\subset  \CC[x_1,\dots, x_n]$; hence we may also associate a variety $V(I)=V(f_1,\dots, f_r)$ to an ideal $I=\langle f_1,\dots, f_r\rangle$.

As in \S\ref{sec:CRNT} we will consider a chemical reaction network $\mathcal{G}=(\mathcal{C},\mathcal{R})$ defining a ODE model $$
  \frac{dx}{dt} = \alpha^T A_\kappa x^\alpha
$$as in \eqref{eq:DynSys}. The set of all steady states of this system of ODEs (in an affine space $\CC^n$) defines an algebraic variety generated by the {\em steady state ideal} $I = \langle \alpha^T A_\kappa x^\alpha\rangle \subseteq R = \CC[x_1,\dots,x_n]$. The ring $R$ is generated by the chemical concentrations and defined over the field $\CC$.

A toric variety is an algebraic variety which can be described as the image of a monomial map. Work over the field $\CC$ and 
fix an integer $d \times n$-matrix
 $A$, with columns $a_1,a_2,\ldots,a_n$, 
and rank $d$.
A given column vector $a_i$ defines a (Laurent) monomial
$t^{a_i} = t_1^{a_{1i}} t_2^{a_{2i}} \cdots t_d^{a_{di}}$ where $t \in (\CC^*)^d$, and $\CC^* = \CC \backslash \{0\}$ denotes the non-zero elements of the field (this is sometimes called the complex torus). 
The {\em affine toric variety over $\CC$} defined by $A$ is $$X_A\,=\,\overline{\{ (t^{a_1}, \ldots , t^{a_n}) \,:\,t \in (\CC^*)^d \}}\subset \CC^n.$$ In words, the variety $X_A\,$ is the (Zariski) closure in $\CC^n$ of the image of the monomial map defined by $A$. The defining implicit equations for $X_A$ can always be chosen to be binomials, that is $X_A=V(I)$ where I is an ideal defined by binomial equations. Specifically by \cite[Corrollary~4.3]{GBCP} we have that \begin{equation}
I=\left\langle x^{c^+}-x^{c^-}\; | \; c \in \ker(A) \right\rangle ,\label{eq:BinomialIdeal}
\end{equation}
 where $c^+_i$ is equal to $c_i$ if $c_i>0$ and $0$ otherwise, and where $c^-_i$ is equal to $|c_i|$ if $c_i<0$ and $0$ otherwise. The ideal $I$ above is, additionally, a {\em prime ideal}; we say an ideal $I\subset \CC[x_1,\dots, x_n]$ is prime if whenever $f\cdot g \in I$ then either $f\in I$ or $g\in I$ (or both are in $I$). Geometrically this means the associated variety is irreducible (i.e.~it cannot be written as a non-trivial union of other varieties) and reduced (meaning a generic point has multiplicity one, e.g. the roots of $x^2=0$ have multiplicity two). 
 
 Conversely, any prime polynomial ideal $\langle f_1,\dots,f_m \rangle$ in $\CC[x_1,\dots,x_{n}]$, where each $f_i$ is a binomial, defines a toric variety ${X}_A$ in $\CC^n$.
Chemical reaction networks whose steady state ideal is generated by such prime ideals have been studied in the literature e.g.~\cite{PerezMillan}. To simplify notation we make the following definition.

\begin{definition}
A chemical reaction network whose steady state ideal has an associated prime that is binomial is a toric chemical reaction network.
\label{def:ToricNetwork}
\end{definition}

The fact that toric varieties can be parameterized by monomials over $\CC^*$ also yields a parametric representation of the positive part of the toric variety, i.e.~of the positive steady states of the associated reaction network. For this reason we restrict our analysis to toric networks. We will also wish to explicitly incorporate the reaction rates in our analysis, hence our starting point will be a binomial ideal with coefficients in $\RR(\kappa):=\RR(\kappa_1,\dots ,\kappa_m)$, the field of rational functions of the reaction rates with real coefficients. 
More precisely we consider a prime binomial steady state ideal defined by $\nu$ equations, $$I =  \langle c_ix^{b^+_i} - k_ix^{b^-_i}\; | \; i=\{1,\dots,\nu\}, \;c_i,\:k_i\in \RR(\kappa) \rangle\subseteq \RR(\kappa)[x_1,\dots,x_n].$$ The vectors $b^+_i$ and $b^-_i$ are positive column vectors with disjoint support and let $B=(b_1,\dots,b_\nu)$ be the matrix with columns $b_i = b^+_i - b^-_i$, called the exponent matrix of $I$. Then there exists a matrix $A$ such that $b_i \in \ker(A),$ for $i= 1,\dots,\nu$. The matrix $A$ is a $d\times n$ integer matrix where $d$ is the dimension of the toric variety and $n$ is the dimension of the ambient affine space. Note that, since the dimension
of a variety cannot exceed the dimension of its ambient space, $\text{rank}(A) \leq d \leq n$.

\begin{remark}
Note that there is some freedom in the choice of the matrix A, as integer row operations on rows of A will not change its kernel and, hence, the variety it parameterizes is left unaltered. Therefore, we informally refer to any matrix $A$ as \emph{the} $A$-matrix (associated to a toric variety $X_A=V(I)$).
\end{remark}

To explicitly incorporate the reaction rates in the monomial parameterization we will now define the scaled toric variety of a matrix $A$ and a vector $x^*\in \CC^n$. The value $x^*$ can be thought of a particular steady state which in turn depends on the reaction rates, see \mbox{Remark \ref{rem:Constants}}.

\begin{definition}[Scaled Toric Variety]
Given an $x^*=(x_1^*,\dots,x_n^*)\in \CC^n$ define a (Laurent) monomial map $\psi_A^{(x^*)}:=\psi_A$ with
$$
  \psi_A:(\CC^*)^d\to \CC^n {\;\; \;\rm where}\;\;
  t\mapsto (x^*_1t^{a_1},\dots, x^*_nt^{a_n})
$$
for $a_i$ a column of $A$. The (Zariski) closure of the image of this monomial map defines a scaled toric variety,
$X_{A,x^*}=\overline{\psi_A((\CC^*)^d)}$.
\label{def:ToricVariety}
\end{definition}

The monomial map $\psi_A$ also induces a parameterization map.
\begin{definition}
The parameterization map defined by the $A$-matrix is the $\CC$-algebra homomorphism
\begin{align*}
   \phi_A: \CC[x_1,\dots,x_n] &\to \CC[t_1^\pm,\dots,t_d^\pm]\\
   \phi_A(x_i) = x^*_it^{a_i} &= x^*_i t_1^{a_{i_1}}\cdots t_d^{a_{i_d}}.
\end{align*}
Note, as above, the map $\phi_A$ depends on $x^*\in \CC^n$.
\label{def:ParamMap}
\end{definition}

\begin{example}
The steady state ideal of \eqref{netw:Phos1} is 
\begin{align*}
    I = \langle &-\kappa_1x_1x_3+(\kappa_2+\kappa_3)x_5,-\kappa_4x_2x_4+(\kappa_5+\kappa_6)x_6,\\
    &-\kappa_1x_1x_3+\kappa_2x_5+\kappa_6x_3,\kappa_3x_5-\kappa_4x_2x_4+\kappa_5x_6,\\
    &\kappa_1x_1x_3-(\kappa_2+\kappa_3)x_5,\kappa_4x_2x_4-(\kappa_5+\kappa_6)x_6\rangle,\\
    =\langle &\kappa_2x_5-\kappa_5x_6, \kappa_4x_2x_4 -(\kappa_5+\kappa_6)x_6, \kappa_1x_1x_3 - (\kappa_2+\kappa_3)x_5\rangle, 
\end{align*}
which is a (prime) binomial ideal in the ring $\RR(\kappa_1,\dots,\kappa_6)[x_1,\dots,x_6]$. Hence, we can find a parameterization $x_1 = x_1^*t_1,\: x_2 = x_2^*t_2,\: x_3 = x_3^* t_1^{-1}t_3^{-1},\: x_4 = x_4^*t_2^{-1}t_3^{-1},\: x_5 = x_5^* t_3^{-1},\: x_6 = x_6^* t_3^{-1}$ corresponding to the A-matrix $$A=\begin{pmatrix}
1 & 0 & -1 & 0 & 0 & 0\\
0 & 1 & 0 & -1 & 0 & 0\\
0 & 0 & -1 & -1 & -1 & -1
\end{pmatrix}.$$

\label{ex:Simple}
\end{example}

\begin{remark}
As shown in \cite{PerezMillan} the entries of $x^*$ are algebraic functions in $\kappa$. In practice we wish to restrict our choices of $\kappa$ to those such that when we evaluate $x^*$ at some fixed $(\kappa_1, \dots, \kappa_m)\in \RR_{>0}^m$ we have that $x^*\in  \RR_{>0}^n$. 
\label{rem:Constants}
\end{remark}
The proposition below connects the number of conserved quantities to the dimension of the toric variety and is straightforward, however, we include a proof as we were unable to locate one in the literature. 
 
\begin{proposition}
Consider a chemical reaction network with steady state ideal $I = \langle\alpha^TA_k x^\alpha\rangle$ in the polynomial ring $\RR[x_1,\dots,x_n]$. Let $V=\overline{V(I)\cap (\RR^*)^n}$ be the real steady state variety and denote the dimension of $V$ by $\dim(V)$. 
Let the linear equations $\ell_1,\dots, \ell_d$ be the corresponding conservation relations. Suppose that $\dim ({V\cap V(\ell_1,\dots \ell_d)})=d'$. Then $\dim(V)=d+d'$. Further if $V$ is a toric variety we have that $\dim(V\cap (\RR_{>0})^n)=d+d'$.\label{prop:DimAndConserveRelations}
\end{proposition}\begin{proof}
Each linear form $\ell_j$ contains a constant term $c_j\in \RR$ which can be chosen freely. For a sufficiently general choice of $c_1,\dots, c_d$ the intersection $V\cap V(\ell_1,\dots \ell_d)$ is transverse, it follows that $d'=\dim(V\cap V(\ell_1,\dots \ell_d))=\dim(V)-\dim(V(\ell_1,\dots \ell_d))=\dim(V)-d$. Because $V$ is toric it is parameterized by monomials, and hence its dimension in the positive orthant is the same as its dimension over $\RR^n$.

\end{proof}

\begin{remark}
In the notation of Proposition \ref{prop:DimAndConserveRelations} any $d'>0$ implies that there is an infinite number of steady states for each general choice of $c_1,\dots, c_d$. Hence, the case $d' > 0$ renders the discussion of multistationarity irrelevant. Therefore for the remainder of this paper we assume that $d'=0$ meaning that the dimension of the toric variety is equal to the number of conservation relations.
\label{rem:assump}
\end{remark}

\begin{example}[Example \ref{ex:Simple} cont.]
The network of Example \ref{ex:Simple} it is known to have a finite number of positive steady states for each set of initial conditions ($d'=0$) and there are three independent conservation relations $x_1+x_5 = c_1$, $x_2+x_6 = c_2$ and $x_3+x_4+x_5+x_6 = c_3$. This implies that the dimension of the (toric) steady state variety, $V(I)$, is $\dim(V(I)) = 3$ which is indeed the case.
\end{example}

\subsection{Matroids}

One of the focuses of this paper is to find and study independent subsets of chemical species. Independent subsets can give valuable information about which species concentrations have to be measured and which concentrations can be determined from measurements \cite{Gross2016}. A simple example of linear independence is the set of three vectors $v_1 = (1,0)^T,\: v_2=(0,1)^T,\: v_3=(2,1)^T$. The set of vectors $\{v_1,v_2\}$ is linearly \emph{independent} as there exists no $\lambda \in \RR$ such that $\lambda v_1 = v_2$, whereas $v_3$ can be obtained by $v_1+v_2$ and $\{v_1,v_2,v_3\}$ is therefore \emph{dependent}. The vectors $v_1$ and $v_2$ are said to from a \emph{basis} of the set $\{v_1, \: v_2,\: v_3\}$ which is a familiar notion of linear algebra. The set $\{v_1,\:v_2,\: v_3\}$ is minimally dependent as it contains only a single element other than a basis and is a \emph{circuit}. Matroid theory extends the notion of independence to polynomials rings \cite{Oxley}. First, we define a matroid.

\begin{definition}[Matroid]
A matroid $\mathcal{M}(E,\mathcal{I})$ is a pair of two finite sets $(E,\mathcal{I})$ where $E$ is the ground set and $\mathcal{I}$ is a set of subsets of $E$, called independent sets, satisfying the following conditions.
\begin{enumerate}
\item The empty set, $\emptyset$, is independent such that $\emptyset\in \mathcal{I}$ and, hence, $\mathcal{I} \neq \emptyset$.
\item If $i\in\mathcal{I}$ and $i'\subseteq i$ then $i'\in\mathcal{I}$. This is called the {\em hereditary property}.
\item If $i_1,\: i_2 \in\mathcal{I}$ and $|i_1|<|i_2|$ then there exists an element $x\in i_2-  i_1$ such that $i_1\cup x \in \mathcal{I}$. This is the {\em exchange property}.
\end{enumerate}
\end{definition}

The notions of matroid \emph{basis}, \emph{rank} and \emph{circuit} will also be useful.

\begin{definition}
Matroid bases, rank and circuits are defined as follows:
\begin{itemize}
\item A \emph{basis} $S$ of a matroid $\mathcal{M}$ is a maximally independent subset, i.e.~a subset $S\in\mathcal{I}$ such that $S\cup k$ is a dependent set for all $k\in E-S$. Define the set of all bases to be $\mathcal{B}$.
\item The \emph{rank} $\rho$ is a function on a subset of $E$ which takes a set $e \subseteq E$, and returns the cardinality of the largest subset $i\subseteq e$ which also satisfies $i\in\mathcal{I}$.
\item The \emph{circuits} $C$ are minimally dependent sets, i.e.~subsets of $E$ which satisfy $C\not\in \mathcal{I}$ and $(C-z)\in\mathcal{I}$ for any $z\in C$.
\end{itemize}
\end{definition}

In particular for a matroid $\mathcal{M}$ it holds that all bases of $\mathcal{M}$ are equicardinal. An important notion in matroid theory is also the rank of a matroid.
\begin{definition}
The \emph{rank} of a matroid $\mathcal{M}$, $\rho(\mathcal{M})$, is the cardinality of a basis.
\end{definition}

Due to the general definition of a matroid many mathematical objects have a matroid structure such as sets of vectors in a vector space or graphs. One class of matroids relevant for chemical reaction networks are {\em algebraic matroids}. Algebraic matroids encode the algebraic dependencies between the variables of a polynomial ring $\CC[x_1,\dots,x_n]$ in a prime ideal $P$. 
A more detailed discussion of these ideas in given in Section \ref{sec:matroids}. 
For information on how to compute algebraic matroids we refer the reader to \cite{Oxley,Rosen2014}.
Since we are interested in the relation between the matroids of families of reaction networks we introduce the notion of a submatroid.

\begin{definition}[Submatroid]
Given a subset $E'\subset E$, a matroid $\mathcal{M}' (E',\mathcal{I}')$ is a submatroid of $\mathcal{M}(E,\mathcal{I})$ if $\mathcal{I}' = \mathcal{I}\cap\mathcal{P}(E')$, where $\mathcal{P}(E')$ is the power set of $E'$. The rank function of $\mathcal{M}'$ is the restriction $\rho_{\mathcal{M}'}(S) = \rho_\mathcal{M}(S)$ for $S\subseteq E'$.
\end{definition}
Submatroids of algebraic matroids contain a subset of the polynomial relations between the variables of circuits of the original matroid. The relations between the variables of circuits are referred to as {\em circuit polynomials}. Below we will consider algebraic matroids of toric networks and the relations of matroids within a family of networks.

\begin{example}
The algebraic matroid of \ref{ex:Simple} can be calculated using the computer algebra system Macaulay2 \cite{M2}. It is a matroid on the ground set $E = \{x_1,\dots,x_6\}$ of rank $3$ with $12$ bases. Hence, as we will show in later sections, measuring only four species will suffice to reject the one site model, any three species in a basis and one extra species. Further, the matroid of the two site distributive system is a rank $3$ matroid with $61$ bases and the matroid of the two site processive network is a rank $3$ matroid with $19$ bases. Both two-site model matroids contain the one site matroid as a submatroid.
\end{example}

\section{Families of Reaction Networks}\label{sec:FamiliesOfGraphs}

In this paper we obtain results which hold for a range of networks rather than a single network. When the properties studied (i.e.~multistationarity) are present in a ``small network'', these properties lift to all larger networks constructed by the procedure outlined below. We call the collection of all such networks a family $\mathcal{N}$ with members $\mathcal{N}_i$. In this section we rigorously define families of networks and give examples of chemical systems fulfilling the family property. Furthermore, we prove some results on families which have toric steady states.

\subsection{Family Definition and Properties}

Families of networks can be found in a wide range of biochemical settings such as multi-site phosphorylation \cite{ProcPhos} (e.g.~cellular signalling, DNA transcription, cell death), kinetic proofreading \cite{McKeithan1995} (immune response) or compartmentalised diffusion \cite{Woolley2011} (spatial models).

We identify a family by properties of its reaction graphs. If we can construct the graph of another network from a given network by a fixed set of procedures, the networks are in the same family; an example of such a construction is given in Figure \ref{fig:Example}.  

\begin{figure}[h!]
    \centering
    \subfloat[The graph $\mathcal{G}_0$.]{\includegraphics[width=.25
    \textwidth]{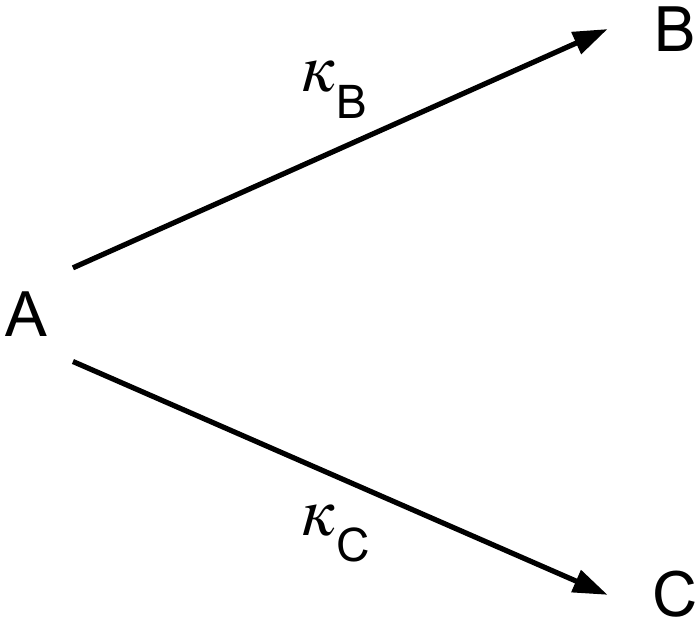}\label{subfig:1}}\hspace{1cm}
    \subfloat[Adding a species $Y_1$ to form the intermediate network]{\includegraphics[width=.25\textwidth]{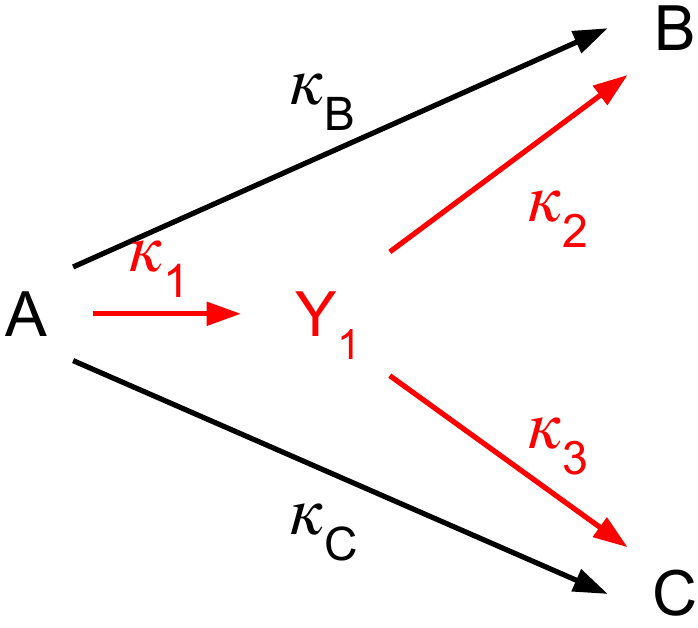}\label{subfig:2}}\hspace{1cm}
    \subfloat[Deleting edges to form the graph $\mathcal{G}_1$]{\includegraphics[width=.25\textwidth]{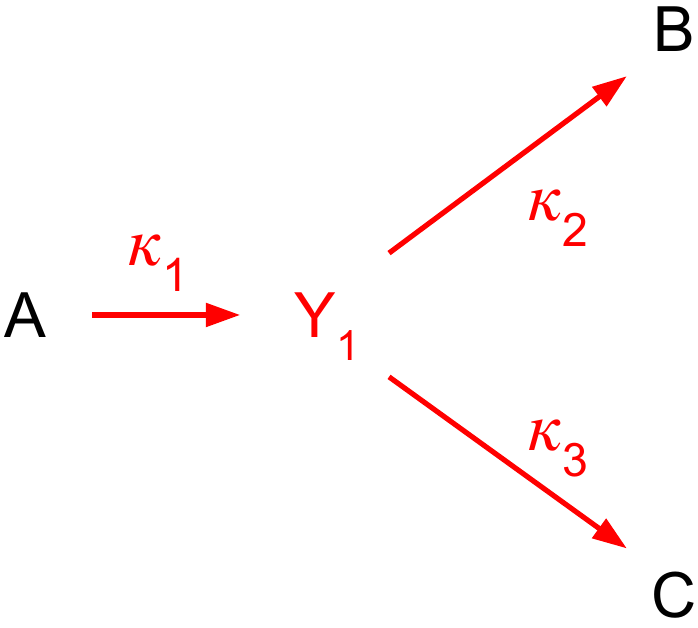}\label{subfig:3}}
    \caption{Starting from the graph in \ref{subfig:1}, first new vertices and edges are added in \ref{subfig:2} and, second, some edges of the original graph are deleted in \ref{subfig:3}. Note, that Figures \ref{subfig:1} and \ref{subfig:3} can be obtained from the intermediate network \ref{subfig:2} by setting the edge weights $\{\kappa_1,\kappa_2,\kappa_3\}$ or, respectively, $\{\kappa_A,\kappa_B\}$ to zero.}
    \label{fig:Example}
\end{figure}

\begin{remark}[Informal Construction of Families of Graphs]
Fix a labelled digraph $\mathcal{G}_n$ and an unlabelled digraph $M$; from these build new family members step-by-step. At each step assign a new set of labels to $M$, rendering $M$ a labelled digraph. The vertices of $\mathcal{G}_{n+1}$ are the union of those of $\mathcal{G}_{n}$ and of $M$. The edges of the graph $\mathcal{G}_{n+1}$ are obtained from $\mathcal{G}_{n}$ as follows:
\begin{enumerate}
    \item Add a set $\mathcal{E}$ of edges adjacent to both some vertices of $M$ and some vertices of $\mathcal{G}_n$; the resulting graph $\mathcal{G}_n^{\rm int}$ is referred to as the {\em intermediate graph} (see also Definition \ref{def:InterNet}).
    \item Delete some subset of the edges of $\mathcal{G}_{n}$ from the graph $\mathcal{G}_n^{\rm int}$ to form $\mathcal{G}_{n+1}$.
\end{enumerate}
\label{remark:InformalFamilyConstuct}
\end{remark}

We formalise this construction using a definition from graph theory \cite{Noy2004}. Graphs constructed as in Remark \ref{remark:InformalFamilyConstuct} form a family if they satisfy the conditions of Definition \ref{def:recursiveG} below.

\begin{definition}[See also Section 2 of \cite{Noy2004}]
For a reaction graph $\mathcal{G}=(\mathcal{C},\mathcal{R})$ and for $U\subseteq \mathcal{C}$ let $\mathcal{G}[U]$ be the subgraph induced by $U$ and $N_\mathcal{G}(U)$ the set of vertices adjacent to some vertex in $U$. Consider a set of graphs $\{\mathcal{G}_i\}_{i\geq 0}$ where $\mathcal{G}_i=(\mathcal{C}_i,\mathcal{R}_i)$. Set $\mathcal{W}_0=\mathcal{C}_{0}$ and $\mathcal{E}_0=\mathcal{R}_0$; additionally let $\mathcal{W}_{i+1} = \mathcal{C}_{i+1} - \mathcal{C}_{i},$ and $\mathcal{E}_{i+1} = \mathcal{R}_{i+1}-\mathcal{R}_{i}$ for $i\geq 0$. Further, there exists a positive integer $n$, a
labelled graph $M$ and a set of reactions $Y$, such that the set $\{\mathcal{G}_i\}_{i\geq 0}$ is called a family of graphs if the following properties hold:
\begin{enumerate}
\item $N_{\mathcal{G}_n}(\mathcal{W}_n)\subseteq \mathcal{W}_0\cup \left(\bigcup_{i=0}^r \mathcal{W}_{n-i}\right)$ for $n>r$,
\item $\mathcal{R}_n=\left(\mathcal{R}_{n-1} - Y\right)\cup \mathcal{E}_n$, where $Y\subseteq \bigcup_{i=1}^r \mathcal{E}_{n-i}$,
\item the graph $\mathcal{G}_n\left[\mathcal{W}_0\cup\left(\bigcup_{i=0}^r \mathcal{W}_{n-i}\right)\right]$ is equal to $M$ for $n>r$. In addition, $\mathcal{G}_n[\mathcal{W}_n]$ is always the same graph.
\end{enumerate}
\label{def:recursiveG}
\end{definition}
Put simply the last condition says that the graph ``added" to the previous graph must be the same for each step throughout the family; this is illustrated in Figure \ref{fig:LadderGraph} and further elucidated in \cite{Noy2004}.
\begin{figure}
    \centering
    \includegraphics[width=.4\textwidth]{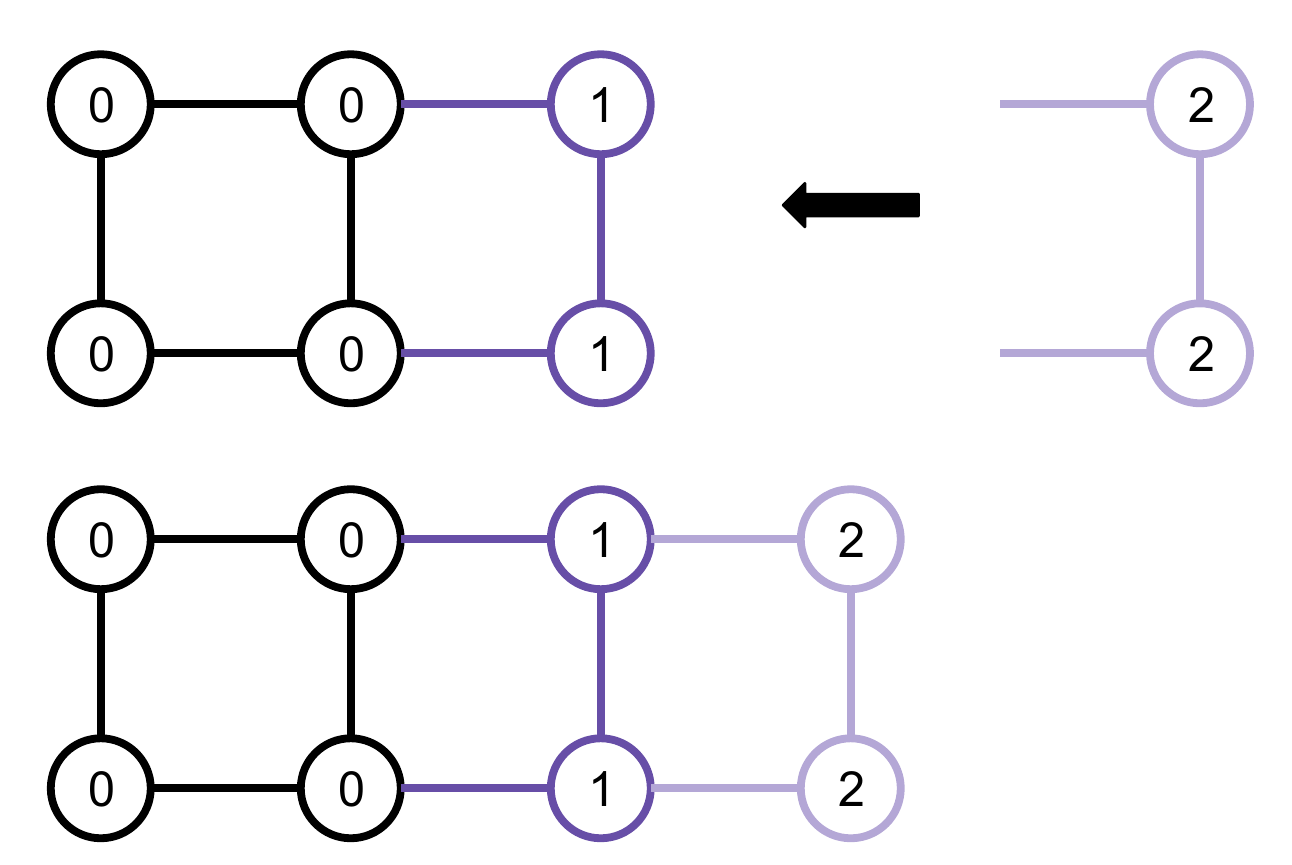}
    \caption{An illustration of building a two-site ladder graph. The vertices are labelled according to the sets $\mathcal{W}_i$ and the edge colours highlight the sets $\mathcal{E}_i$. It can be seen that the labelled graph $M$ added to $\mathcal{G}_i$ is always the same.}
    \label{fig:LadderGraph}
\end{figure}
In the context of chemical reaction networks we add another condition, namely, that in every subsequent family member there appears at least one new chemical species:
\begin{enumerate}
    \item[\textit{4.}] $\mathcal{S}_M\subsetneq \mathcal{S}_N$\text{ for }$M<N$.
\end{enumerate}
We assume that in a network with a given set of chemical species every complex which  is chemically possible and that can be formed with this set of species is already present in the network; i.~e.~if in a network with species $H_2$ and $O_2$ chemistry dictates that the only possible reaction is $2\,H_2 + O_2 \to 2\,H_2O$ then only the complexes $2\,H_2 + O_2$ and $2\,H_2O$ exist for this species set. Therefore, to give rise to a larger network new species have to be added.

\begin{remark}
Note that not every infinite sequence of networks is a family as defined in Definition \ref{def:recursiveG}. For a non-example see Example \ref{ex:NonEx}.
\label{rem:Warning1}
\end{remark}

\begin{remark}
For ease of notation, the unlabelled graph $M$ is drawn in a ``dot 
and arrow'' representation; e.g. the graph $M = \{\boldsymbol{\cdot}, \boldsymbol{\cdot} \to \boldsymbol{\cdot}\}$ corresponds to a graph with two connected components, an isolated vertex and a component with two unlabelled vertices and a directed edge between them.
\end{remark}

\begin{example}[Distributive Phosphorylation]
Distributive phosphorylation with $N=1$ phosphorylation sites follows the reaction scheme
\begin{align*}
  E+S_0 &\rightleftharpoons ES_0\rightarrow E+S_1,\\
  F+S_1 &\rightleftharpoons FS_1\rightarrow F+S_0.
\end{align*}
The reaction scheme is a digraph $\mathcal{G}_1 = (\mathcal{W}_1,\mathcal{E}_1)$ with $\mathcal{W}_1 = \{E+S_0,ES_0,E+S_1,F+S_1,FS_1,F+S_0\}$ and $\mathcal{E}_1 = \{(E+S_0, ES_0),(ES_0, E+S_0), (ES_0, E+S_1), (F+S_1, FS_1), (FS_1, F+S_1), (FS_1, F+S_0)\}$.
To construct $\mathcal{G}_2$ we use $\color{OliveGreen}\mathcal{W}_2=\{ES_1,E+S_2,F+S_2,FS_2\}$ and $\color{OliveGreen}\mathcal{E}_2=\{(E+S_1, ES_1), (ES_1, E+S_1),(ES_1, E+S_2),(F+S_2, FS_2),(FS_2, F+S_2),(FS_2, F+S_1)\}$. Defining $\mathcal{G}_2 =(\mathcal{W}_1\cup \mathcal{W}_2,\mathcal{E}_1\cup \mathcal{E}_2)$ gives
\begin{align*}
  E+S_0 &\rightleftharpoons ES_0\rightarrow E+S_1 {\color{OliveGreen}\rightleftharpoons ES_1\rightarrow E+S_2},\\
  {\color{OliveGreen}F+S_2\rightleftharpoons FS_2\rightarrow}F+S_1 &\rightleftharpoons FS_1\rightarrow F+S_0.
\end{align*}
Hence (in the notation of Definition \ref{def:recursiveG}) , $M = \{\boldsymbol{\cdot}\to \boldsymbol{\cdot},\: \boldsymbol{\cdot}\rightleftharpoons \boldsymbol{\cdot}\}$ and $r=1$. The digraphs $\mathcal{G}_1$ and $\mathcal{G}_2$ define the family members $\mathcal{N}_1$ and $\mathcal{N}_2$.
Inductively, this procedure can be continued to member $\mathcal{N}_N$.
By inspection of the green subgraph we can see that condition 3 of Definition \ref{def:recursiveG} is fulfilled for any $N>0$ and, therefore, the distributive phosphorylation networks form a family. Further, $\mathcal{N}_1$ is a subnetwork of $\mathcal{N}_2$ as defined in \cite[Definition 2.2]{Joshi2013}.
\label{ex:DistPhos}
\end{example}


\begin{example}[Processive Phosphorylation]
Processive phosphorylation of a substrate with $N=1$ phosphorylation sites follows the reaction scheme
\begin{align*}
  E+S_0 &\rightleftharpoons ES_0\rightarrow E+P,\\
  F+P &\rightleftharpoons FS_1\rightarrow F+S_0.
\end{align*}
For $N=1$ this is the same reaction scheme as for distributive phosphorylation, except for the relabelling of the fully phosphorylated substrate as product $P$.
However, to construct the next family member from the digraph $\mathcal{G}_1$ we use $\color{OliveGreen}\mathcal{W}_2=\{ES_1,FS_2\}$ and $\color{OliveGreen}\mathcal{E}_2=\{(ES_0, ES_1), (ES_1, E+P),(F+P, FS_2), (FS_2, F+P), (FS_2, FS_1)\}$. Next, delete edges $Y_1=\{(ES_0,E+P),(F+P,FS_1),(FS_1,F+P)\}$ which results in the graph $\mathcal{G}_2 =\left(\mathcal{W}_1\cup \mathcal{W}_2,\left(\mathcal{E}_1\cup \mathcal{E}_2\right)- Y_1\right)$ to give
\begin{align*}
  E+S_0 \rightleftharpoons ES_0 &{\color{OliveGreen}\rightarrow ES_1\rightarrow}E+P,\\
  F+P &{\color{OliveGreen}\rightleftharpoons FS_2\rightarrow} FS_1\rightarrow F+S_0.
\end{align*}
Where the vertices $M=\{\boldsymbol{\cdot},\boldsymbol{\cdot}\}$ were added and $r=1$ as in Definition \ref{def:recursiveG}.
Again, the digraphs $\mathcal{G}_1$ and $\mathcal{G}_2$ define the family members $\mathcal{N}_1$ and $\mathcal{N}_2$.
Condition 3 of Definition \ref{def:recursiveG} is fulfilled for any $N>1$ and, therefore, the processive phosphorylation networks form a family. 
\label{ex:ProcPhos}
\end{example}

\begin{example}[Non-example]
As mentioned in Remark \ref{rem:Warning1} not every infinite sequence of graphs forms a family. Consider the autocatalytic networks
\begin{align*}
    &X_i+X_{i+1}\rightarrow 2X_{i+1},\\
    &X_i\rightleftharpoons\emptyset\qquad \text{for }i=1,\dots, N\text{ and setting } X_{N+1}=X_1.
\end{align*}
In the reaction above $\emptyset$ represents production and degradation of a molecule by a mechanism not further studied.
Without loss of generality consider $\mathcal{N}_7$ and $\mathcal{N}_8$; we see that $\mathcal{C}_8$ is not a union of $\mathcal{C}_7$ and another graph as $X_7+X_1\in\mathcal{C}_7\not\subseteq\mathcal{C}_8$.
\label{ex:NonEx}
\end{example}

We conclude this subsection by considering the polynomial equations arising from the reaction graphs of successive members of a family. First, we see from Section \ref{sec:CRNT} that every reaction has a unique rate constant, $\kappa_i$, which is the edge weight of the reaction digraph. Further, in the dynamical system given by the equations \eqref{eq:DynSys} every monomial, which represents a reactant complex, is multiplied by a constant $\kappa_i$. Note that isolated vertices in the reaction graph do not contribute to the dynamics of the network. Hence, for families of networks which satisfy Definition \ref{def:recursiveG} with $Y = \emptyset$, we can derive the equations of the $N^\text{th}$ member of a family from the $(N+1)^\text{th}$ member by setting all the edge weights of the edges added to the reaction graph to zero. Formally, we define the evaluation map
\begin{align}
    \pi:\RR[\kappa_1,\dots,\kappa_{m+m'}][x_1,\dots,x_n]&\to\RR[\kappa_1,\dots,\kappa_m][x_1,\dots,x_n]\nonumber\\
    \pi(\kappa_i) &= \begin{cases}
    \kappa_i\text{ if }i\leq m,\\
    0\text{ otherwise.}
    \end{cases}
\end{align}
For networks which require the deletion of a set of edges, $Y$, we need the concept of an {\em intermediate network}.

\begin{definition}[Intermediate Network]
The intermediate network $\mathcal{G}_n^{\rm int}$ is the network constructed by only adding reactions to join the newly labelled graph $M$ to $\mathcal{G}_n$ and before any edges are deleted from the reaction graph. See step 1 of Remark \ref{remark:InformalFamilyConstuct}. 
\label{def:InterNet}
\end{definition}

\begin{remark}
The next member of a family is a subnetwork of the intermediate network. Hence, for families in which the next family member coincides with the intermediate network (i.e.~when $Y=\emptyset$ in Definition \ref{def:recursiveG}) this terminology is not required.
\end{remark}

The dynamical equations of the $N^\text{th}$ and $(N+1)^\text{th}$ member can be constructed from their intermediate network by defining the appropriate evaluation maps. As can be seen from condition 2 of Definition \ref{def:recursiveG} the reaction (edge) set of the $(N+1)^\text{th}$ member of a family is given by $\mathcal{R}_{N+1} = (\mathcal{R}_N - Y)\cup \mathcal{E}_N$. The edge sets for the $N^\text{th}$ and intermediate network are given by $\mathcal{R}_N$ and $\mathcal{R}_N\cup \mathcal{E}_N$ respectively. Define a function $\mu$ which associates a unique edge weight, $\kappa_{i,j}$, to every edge of the reaction graph,
\begin{align}
    \mu:\; \mathcal{C}\times\mathcal{C} &\to \RR[\kappa],\nonumber\\
    \mu\left((C_i,C_j)\right) &= \kappa_{i,j}.\label{eq:muFunc}
\end{align}
This gives the rate constants $\mu\left(\mathcal{R}_N\right) = \{\kappa_1,\dots,\kappa_m\}$, $\mu\left(\mathcal{E}_N\right) = \{\kappa_{m+1},\dots,\kappa_{m+m'}\}$ and $\mu\left(Y\right) = \{\kappa_1,\dots,\kappa_{m''}\}$ (where we relabel rate constants produced by \eqref{eq:muFunc} as required). By construction, an edge is not present in the reaction graph if it has an edge weight of $0$. Hence, two evaluation maps can be defined to map the edge sets of the intermediate network to the $(N+1)^\text{th}$ and $N^\text{th}$ network respectively,
\begin{align}
    \pi^+:\RR[\kappa_1,\dots,\kappa_{m+m'}][x_1,\dots,x_n]&\to\RR[\kappa_{m''+1},\dots,\kappa_{m+m'}][x_1,\dots,x_n]\nonumber\\
    \pi(\kappa_i) &= \begin{cases}
    \kappa_i\text{ if }i> m'',\\
    0\text{ otherwise,}
    \end{cases}
\intertext{and}
    \pi^-:\RR[\kappa_1,\dots,\kappa_{m+m'}][x_1,\dots,x_n]&\to\RR[\kappa_1,\dots,\kappa_m][x_1,\dots,x_n]\nonumber\\
    \pi(\kappa_i) &= \begin{cases}
    \kappa_i\text{ if }i\leq m,\\
    0\text{ otherwise.}
    \end{cases}\label{eq:pim}
\end{align}

This results in the edge weights $\pi^+\left(\mu\left(\mathcal{R}_N\cup \mathcal{E}_N\right)\right) = \{\kappa_{m''+1},\dots,\kappa_{m+m'}\} = \mu\left((\mathcal{R}_N - Y)\cup \mathcal{E}_N\right)$ and $\pi^-\left(\mu\left(\mathcal{R}_N\cup \mathcal{E}_N\right)\right) = \{\kappa_1,\dots,\kappa_{m}\} = \mu\left(\mathcal{R}_N\right)$, which, after taking the inverse of $\mu$, can be associated to the reaction sets of the $(N+1)^\text{th}$ and $N^\text{th}$ network respectively. In particular, if $\mathcal{G}_N^{\rm int}$ is the intermediate graph associated to $\mathcal{G}_N$ then applying $\pi^-$ to (the edge set of) $\mathcal{G}_N^{\rm int}$ will yield $\mathcal{G}_N$ and applying $\pi^+$ to (the edge set of) $\mathcal{G}_N^{\rm int}$ will yield $\mathcal{G}_{N+1}$.

\begin{example}[Processive Phosphorylation]
\begin{align}
    &\schemestart
    $E+S_0$\arrow{<=>[\footnotesize $\kappa_4$][\footnotesize $\kappa_5$]}[,0.8]$ES_0$\arrow(a--b){->[\footnotesize {\color{red}$\kappa_3$}]}[,0.8,red] $E+P$  \arrow(@a--c){->[\footnotesize ${\color{OliveGreen}\kappa_7}$]}[45,0.8,OliveGreen]${\color{OliveGreen}ES_1}$\arrow(@c--@b){->[\footnotesize ${\color{OliveGreen}\kappa_8}$]}[-45,0.8,OliveGreen]
    \schemestop\nonumber\\
    &\schemestart
    $F+P$\arrow(a--b){<=>[\footnotesize ${\color{red}\kappa_1}$][\footnotesize ${\color{red}\kappa_2}$]}[,0.8,red] $FS_1$  \arrow(@a--c){<=>[\footnotesize ${\color{OliveGreen}\kappa_9}$][\footnotesize ${\color{OliveGreen}\kappa_{10}}$]}[-45,0.8,OliveGreen]${\color{OliveGreen}FS_2}$\arrow(@c--@b){->[\footnotesize ${\color{OliveGreen}\kappa_{11}}$]}[45,0.8,OliveGreen]\arrow(@b--){->[\footnotesize $\kappa_6$]}[,0.8] $F+S_0$
    \schemestop
\end{align}
The reaction scheme for the first intermediate network of the processive phosphorylation family. If $\kappa_1 = \kappa_2 = \kappa_3 = 0$ then the $N=2$ site network is obtained. Equivalently, if we let $\kappa_7 = \kappa_8 = \kappa_9 = \kappa_{10} = \kappa_{11} = 0$, then we obtain the $N=1$ site network. The corresponding steady state ideals can be obtained in this manner from the ideal from the intermediate network.
\label{ex:ProcIntNet}
\end{example}

\subsection{Toric Families}

As described in Definition \ref{def:ToricNetwork}, a reaction network is called toric if it has toric steady states.  
Showing whether a chemical reaction network is toric is a non-trivial task; previous results have used deficiency theory \cite{Craciun2009}, network structure based methods \cite{PerezMillan} or methods based on detecting binomiality \cite{Conradi2015}. Formally, one needs to find the associated primes of a steady state ideal and compute a Gr{\"o}bner basis of each. If the reduced Gr{\"o}bner basis is binomial, then the ideal is a prime binomial ideal \cite{Eisenbud1996}. However, as every toric variety has non-zero components, removing any components of the steady state variety contained in the coordinate axes is often a first step towards identifying the toric components.
Algebraically this removal is accomplished by computing the saturation  $I^\infty = I:(x_1\cdots x_n)^\infty$, \cite{CLO}. In this subsection we prove a series of small results regarding the steady state ideals of families of toric networks. 
First, we define a toric family.

\begin{definition}[Toric family]
If every member of a family of networks is a toric chemical reaction network, the family is called a toric family.
\end{definition}

\begin{remark}
Prominent examples of toric families are multisite phosphorylation networks or compartmentalised diffusion networks.
\end{remark}

Lemma \ref{lem:SatEval} states that saturation and the evaluation map, $\pi$, commute. Therefore, the network operation of deleting edges in the reaction graph can be carried out before, or after, finding the prime binomial ideal in the steady state ideal.

\begin{lemma}
Let $\pi^-$ be the evaluation map \eqref{eq:pim} and let $I_{N+1},\;I_N\subseteq \RR[\kappa_1,\dots,\kappa_{m+m'},x_1,\dots,x_{n+n'}]$ be the steady state ideals of the $(N+1)^\text{th}$ and $N^\text{th}$ family member respectively. Then it holds that
	\begin{enumerate}[(i)]
		\item $\pi^-(I_{N+1}):(x_1,\dots,x_{n+n'})^\infty = \pi^-(I_{N+1}:(x_1,\dots,x_{n+n'})^\infty)$,
		\item $\pi^-(I_{N+1}:(x_1,\dots,x_{n+n'})^\infty) = I_N:(x_1,\dots,x_n)^\infty$
	\end{enumerate}
	\label{lem:SatEval}
\end{lemma}
The proof of this lemma uses the the geometric interpretation of the statements of the lemma and can be found in Appendix A.
Suppose that $\mathcal{N}$ is a toric family of reaction networks with the $(N+1)^{th}$ member of the family having toric steady states specified by the binomial ideal $I_{N+1}$. 
Lemma \ref{lem:SatEval} establishes a useful connection between species (corresponding to variables $x$), and reactions (corresponding to $\kappa$), or, equivalently, between vertices and edges of a reaction graph. By the extra Condition 4 of Definition \ref{def:recursiveG} every new edge must either originate or end on a complex containing a new species. Hence, by applying the evaluation map $\pi^-$ an ideal of an intermediate network $I^{\text{int}}_{N}$ ($=I_{N+1}$ for nested subnetworks) is automatically mapped from $\RR[x_1,\dots,x_{n+n'}]$ (corresponding to a network with $n+n'$ chemical species) to an ideal in the ring $\RR[x_1,\dots,x_n]$ (corresponding to a network with $n$ species).

\begin{example}[Processive Phosphorylation (cont.)]
The steady state ideal corresponding the network of Example \ref{ex:ProcIntNet} is
\begin{align}
    \begin{split}
        I = \langle &-\kappa_4x_1x_3+(\kappa_5+{\color{red} \kappa_3})x_5+{\color{OliveGreen} \kappa_8}{\color{OliveGreen} x_7},\\
               &-({\color{red} \kappa_1}+{\color{OliveGreen} \kappa_9})x_2x_4+({\color{red} \kappa_2}+\kappa_6)x_6+{\color{OliveGreen} \kappa_{10}}{\color{OliveGreen} x_8},\\
               &-\kappa_4x_1x_3+\kappa_5x_5+\kappa_6x_6,\\
               &{\color{red} \kappa_3}x_5+{\color{OliveGreen} \kappa_8}{\color{OliveGreen} x_7}-{\color{red} \kappa_1}x_2x_4+{\color{red} \kappa_2}x_6-{\color{OliveGreen} \kappa_9}x_2x_4+{\color{OliveGreen} \kappa_{10}}{\color{OliveGreen} x_8},\\
               &\kappa_4x_1x_3-(\kappa_5+{\color{OliveGreen} \kappa_7}+{\color{red} \kappa_3})x_5,\\
               &{\color{red} \kappa_1}x_2x_4+{\color{OliveGreen} \kappa_{11}}{\color{OliveGreen} x_8}-({\color{red} \kappa_2}+\kappa_6)x_6,\\
               &{\color{OliveGreen} \kappa_7}x_5-{\color{OliveGreen} \kappa_8}{\color{OliveGreen} {\color{OliveGreen} x_7}},\\
               &{\color{OliveGreen} \kappa_9}x_2x_4-({\color{OliveGreen} \kappa_{10}}+{\color{OliveGreen} \kappa_{11}}){\color{OliveGreen} x_8} \rangle \subseteq \RR(\kappa_1,\dots,\kappa_{11})[x_1,\dots,x_8].
    \end{split}
\end{align}
As can be seen using the evaluation map $\pi^-$, which sets $\kappa_7 = \cdots = \kappa_{11} = 0$, the variables $x_7$ and $x_8$ disappear from the ideal. Further, computing the saturation $I:(x_1\cdots x_8)^\infty$ and the appropriate evaluation maps we can show that the relations of Lemma \ref{lem:SatEval} hold using Macaulay2 \cite{M2}.
\label{ex:IntProcId}
\end{example}


Proposition \ref{propn:BinId} below shows that the image of a binomial ideal under the evaluation map $\pi$ is either the constant ideal (corresponding to an empty variety) or another binomial ideal. The proof can be found in Appendix A.

\begin{proposition}
Let $R = \RR(\kappa_1,\dots,\kappa_{m+m'})[x_1,\dots,x_{n+n'}]$ and define the partial evaluation map $\pi: R \to \RR(\kappa_1,\dots,\kappa_{m})[x_1,\dots,x_{n+n'}]$.
Suppose that $I_{N+1}:(x_1\cdots x_{n+n'})^\infty\subset R$ is a binomial ideal. Then $I_N=\pi(I_{N+1}):(x_1\cdots x_{n})^\infty\subset\RR(\kappa_1,\dots,\kappa_{m})[x_1,\dots,x_n]$ is either a binomial ideal or the ideal $(1)$. Further, for a fixed choice of rate constants $I_{N+1}:(x_1\cdots x_{n+n'})^\infty\subset\RR[x_1,\dots,x_{n+n'}]$. Then if $V(I_{N+1}:(x_1\cdots x_{n+n'})^\infty)\subset \RR^n$ is a toric variety and if $V(I_N:(x_1\cdots x_n)^\infty)\cap \RR_{>0}^n\neq \emptyset$ then $\overline{V(I_N)\cap \RR_{>0}^n}$ is a toric variety.
\label{propn:BinId}
\end{proposition}

Hence, by Proposition \ref{propn:BinId} if the intermediate network of the $(N+1)^\text{th}$ and the $N^\text{th}$ member of a family is toric then both, the $N^\text{th}$ and $(N+1)^\text{th}$ members are either toric or have empty positive steady state varieties. 

\begin{example}[Example \ref{ex:IntProcId} cont.]
Explicit computations is Macaulay2 show that the ideal $I:(x_1\cdots x_8)^\infty$ is toric. We also know that the processive phosphorylation network has positive steady states. Hence, the positive steady states of the $N=1$ and $N=2$ networks are cut out by toric varieties.
\end{example}

A relationship between the $A$-matrices of successive members of toric families is now found by considering the exponent matrices of the binomial ideals and their corresponding $A$-matrices.

\begin{theorem}
Let $I_{N+1}:(x_1\cdots x_{n+n'})^\infty \subset \RR(\kappa_1,\dots,\kappa_{m+m'})[x_1,\dots,x_{n+n'}]$ be a binomial ideal defining the positive steady states of the $(N+1)^\text{th}$ member of a family of reaction networks. Also let $\pi$ be the evaluation map which sends $\kappa_{m+1}=\cdots=\kappa_{m+m'}=0 $ and assume that $I_N=\pi(I_{N+1}) \subset \RR(\kappa_1,\dots,\kappa_m)[x_1,\dots,x_n]$. Then $I_N:(x_1\cdots x_n)^\infty$ is binomial. Further, for a choice of generators, let $B_{N+1}$, $B_N$ be the matrices associated to the exponents of the binomial ideals $I_{N+1}$ and $I_N$. Similarly let $A_{N+1}$ and $A_N$ be the $A$-matrices corresponding to $B_{N+1}$ and $B_N$. Then:
\begin{enumerate}
\item[(i)] $B_N$ is a submatrix of $B_{N+1}$ and,
\item[(ii)] $A_N$ is a submatrix of $A_{N+1}$, if the submatrix of $B_{N+1}$ formed by the columns which have at least one non-zero entry in the last $n'$ rows has rank at most $n'$.
\end{enumerate}
\label{thm:Submatrix}
\end{theorem}


The proof can be found in Appendix A. An interesting Corollary of the above Theorem occurs if the number of conservation relations is constant.
\begin{corollary}
	Suppose the assumption of Remark \ref{rem:assump} holds and that the number of conservation relations remains constant. Then, the matrix $A_{N}$ is obtained from $A_{N+1}$ by deleting the last $n'$ of its columns.
	\label{cor:ADel}
\end{corollary}
\begin{proof}
	When the dimensions of the varieties of the networks $\mathcal{N}_{N+1}$ and $\mathcal{N}_N$ are the same, then so do the dimensions of the left kernels of their $B$-matrices. Hence, by Theorem \ref{thm:Submatrix} this is the case if and only if $\text{rank}(\tilde{B}) = n'$ and the result follows by the Rouch{\'e}-Capelli theorem.
\end{proof}
This section is concluded by relating the mathematical insights of the above theorems to families of toric chemical reaction networks.

\begin{theorem}
The $A$-matrices of members of a family of chemical reaction networks are submatrices of each other when the assumption of Remark \ref{rem:assump} is valid, all networks considered have the same number of conservation relations and one of the following conditions hold:
\begin{enumerate}
    \item[(i)] The family consists of a sequence of subnetworks with toric steady states.
    \item[(ii)] The family is toric with toric intermediates.
\end{enumerate}
\label{thm:AFamSub}
\end{theorem}
The proof can be found in Appendix A.

\begin{example}[Processive Phosphorylation (cont.)]
The processive network is toric with toric intermediates, every family member has a finite number of positive steady states for each set of initial conditions and the number of conservation relations is constant. The $A$-matrix of the intermediate network of Example \ref{ex:ProcIntNet} is given by
$$
A = \begin{pmatrix}
      {-1}&0&1&0&0&0&{\color{red}0}&{\color{red}0}\\
      0&{-1}&0&1&0&0&{\color{red}0}&{\color{red}0}\\
      1&1&0&0&1&1&{\color{red}1}&{\color{red}1}
      \end{pmatrix},
$$ which is identical to the $A$-matrix of the $N=2$ network. The $A$-matrix of the $N=1$ member is obtained by deleting the last two columns (marked in red) of the $A$-matrix associated to the intermediate network.
\end{example}

\section{Matroid Theory for Toric Chemical Reaction Networks}\label{sec:matroids}

In this section we study biochemical and algebraic question 2 regarding parameter estimation and model rejection. We start by showing the equivalence of two matroids, an algebraic matroid and a much simpler linear matroid. We then decorate the linear matroid with binomials and finally apply these results to model rejection.
First, we prove the that following is a matroid.
\begin{proposition}
Let $E = \{x_1,\dots,x_n\}$, $S = \{x_{i_1},\dots,x_{i_j}\}\subseteq E$ and define $\phi_A(S)$ by $\phi_A(S) = \{\phi_A(x_{i_1}),\dots,\phi_A(x_{i_j})\}$, where $\phi_A$ is the parameterization map of Definition \ref{def:ParamMap}. The ground set $E$ with the set $$\mathcal{I} = \{S\subseteq E:
\text{the monomials in } \phi_A(S) \text{are algebraically independent}\}$$ is a matroid $\mathcal{M}(E,\mathcal{I})$. Further, this matroid is isomorphic to the matroid defined by the column vectors of the $A$-matrix.
\label{propn:PSSmatroid}
\end{proposition}
\begin{proof}
Consider the image of $E$ under $\phi_A$; this is a set of monomials.  A set $S$ of monomials is algebraically independent if and only if there exists no polynomial $p \in k[t^\pm_1,\dots,t^\pm_d]$ such that $p(S) = 0$. This is exactly the condition to give an independent set $i\in \mathcal{I}$. Further, by \cite[Lemma 4.2.10]{Mittmann2013}, a set of monomials is independent if and only if their exponent vectors are linearly independent. Hence, algebraic independence of $\phi_A(S)$ is equivalent to linear independence of the columns of the matrix $A$ corresponding to $S$.
\end{proof}

\begin{definition}
The matroid $\mathcal{M}(A)$ defined by the column vectors, $a_i$, of a full rank integer matrix $A$ encodes the algebraic dependencies of the concentrations at a non-zero (usually positive) steady state of a chemical reaction network and is therefore is called the \emph{positive steady state (PSS) matroid}.\label{def:PSSMatroid}
\end{definition}

The next theorem shows that the PSS matroid and the algebraic matroid defined by the binomial ideal are identical. Therefore, the algebraic matroid can be studied directly using linear algebra operations on the columns of the $A$-matrix defining the PSS matroid. This way bases, circuits and even circuit polynomials can be inferred trivially.

\begin{theorem}
Let $I_b\subseteq R = \CC[x_1,\dots,x_n]$ be a prime binomial ideal defining a toric variety $X_{A,x^*}=V(I_b)$ where $A\in \ZZ^{d\times n}$ is the exponent matrix of the corresponding monomial parameterization. Denote the algebraic matroid defined by $I_b$ as $\mathcal{M}(I_b)$ and the matroid defined by the linear independence of the columns of $A$ by $\mathcal{M}(A)$. Then $\mathcal{M}(I_b) = \mathcal{M}(A)$.

\label{thm:MatroidHom}
\end{theorem}
The proof can be found in Appendix A.
We now proceed to find decorations of the PSS matroid's circuits which reveal the full power of the PPS matroid formulation.

\begin{definition}[Laurent Binomial Associated to a PSS matroid circuit] Let $X_{A,x^*}$ be the toric variety defined by a full rank matrix $A$ and a positive vector $x^*$ as in Definition \ref{def:ToricVariety}. Let $\mathcal{M}(A)$ be the associated  PSS matroid (Definition \ref{def:PSSMatroid}).  Consider a circuit $C = \{a_{i_1},\dots,a_{i_{j-1}}\}\cup \{a_{i_j}\} \subseteq E$ of the matroid $\mathcal{M}(A)$, where $\{a_{i_1},\dots,a_{i_{j-1}}\}\in \mathcal{I}$. We define the Laurent polynomial 
$$
\Phi(C)= x_{i_j}^{\lambda_{i_j}} - \left(x^*_{i_j}\right)^{\lambda_{i_j}}\left(\prod_{l=1}^{j-1} \left(x^*_{i_l}\right)^{-\lambda_{i_l}}\right)\prod_{l=1}^{j-1} x_{i_l}^{\lambda_{i_l}} \in \RR(\kappa_1,\dots,\kappa_m)[x_1^{\pm 1},\dots,x_n^{\pm 1}]
$$ with $\lambda_{i_l}\in \ZZ$ chosen such that $\sum_{l=1}^{j-1}\lambda_{i_l} a_{i_l} = {\lambda_{i_j}}a_{i_j}$ and ${\lambda_{i_j}}$ is positive (this is possible since $C$ is a circuit). The expression $\Phi(C)$ is called the Laurent binomial associated to $C$.
\label{def:LaurentBin}
\end{definition}

\begin{definition}[Binomial Associated to a PSS matroid circuit] Let $\Phi(C)$ be a Laurent binomial of a circuit $C$ of a PSS matroid as in Definition \ref{def:LaurentBin}. The binomial associated to $C$ is $\overline{\Phi(C)}$ which is $\Phi(C)$ with the denominator cleared, i.e. $$\overline{\Phi(C)}= x_{i_j}^{\lambda_{i_j}}x^{\lambda^-} - \left(x^*_{i_j}\right)^{\lambda_{i_j}} \left(\prod_{l=1}^{j-1} \left(x^*_{i_l}\right)^{-\lambda_{i_l}}\right)x^{\lambda^+} \in \RR(\kappa_1,\dots,\kappa_m)[x_1,\dots,x_n],$$ where $\lambda^+_j=\lambda_j$ if $\lambda_j>0$ and zero otherwise and where $\lambda^-_j=|\lambda_j|$ if $\lambda_j<0$ and zero otherwise.
\label{def:PSSBin}
\end{definition}

\begin{example}
For the $N=1$ site processive phosphorylation network the columns $C = \{a_1,a_2,a_3\}\cup\{a_4\}$ form a circuit with a linear relation $a_1-a_2+a_3 = a_4$. The Laurent binomial associated to this circuit is $$\Phi(C) = x_4 - \frac{x_4^*x_2^*}{x_1^*x_3^*}x_1x_3x_2^{-1},$$ and clearing the denominator gives the binomial $$\overline{\Phi(C)} = x_4x_2 - \frac{x_4^*x_2^*}{x_1^*x_3^*}x_1x_3.$$
\end{example}

By the following lemma the (Laurent) binomials associated to a PSS matroid are zero on the steady state variety. Hence, they form a model invariant.
\begin{lemma}
Let $X_{A,x^*}$ be the toric variety defined by a full rank matrix $A$ and a positive vector $x^*$ as in Definition \ref{def:ToricVariety} and let $\mathcal{M}(A)$ be the associated PSS matroid. If $C$ is a circuit in $\mathcal{M}(A)$ then $\Phi(C)(x)=0$ if $x\in (\CC^*)^n \cap X_{A,x^*} $.
\label{lem:parameterization}
\end{lemma}
\begin{proof}
Given a linearly dependent set of vectors it holds that, without loss of generality,  $\sum_{l=1}^{j-1}\lambda_{i_l} a_{i_l} = {\lambda_{i_j}}a_{i_j}$ for some integers $\lambda_{i_l}\in \ZZ$ and ${\lambda_{i_j}}\in \ZZ_{>0}$. It follows that $$
t^{\sum_i^{j-1}\lambda_{i_l}a_{i_l}} = t^{{\lambda_{i_j}}a_{i_j}}.
$$ Taking the preimage under $\phi_A$ this can be rewritten as $$
(x_{i_j}^*)^{\lambda_{i_j}} \prod_{l=1}^{j-1}(x_{i_l})^{\lambda_{i_l}}  \prod_{l=1}^{j-1}(x_{i_l}^*)^{-\lambda_{i_l}}=x_{i_j}^{\lambda_{i_j}}.
$$
\end{proof}
The following theorem states that a set of suitably chosen binomials contains only all positive steady states as our original ideal.
\begin{theorem}
Let $\mathcal{M}(A)$ be the matroid associated to a toric chemical reaction network $\mathfrak{N}$ and choose a basis $S$ and $n-d$ circuits $C_i$ such that $\bigcap_i C_i = S$. Then, the following holds $$V(\overline{\Phi(C_1)},\dots,\overline{\Phi(C_{n-d})})\cap \RR^n_{> 0} = V(I_\mathfrak{N})\cap \RR^n_{> 0}.$$
Hence, proving that the variety $V(\langle\overline{\Phi(C_1)},\dots,\overline{\Phi(C_{n-d})}\rangle)$, when intersected with the subspace spanned by the conservation relations, contains at least two positive points is necessary for proving multistationarity of the original toric network given by $I_\mathfrak{N}$. 
\label{thm:Agreement}
\end{theorem}
\begin{proof}
Let $F_k(x_1,\dots,x_n) = \overline{\Phi(C_k)}$ to give $W=V(F_1,\dots,F_{n-d})$ and also let $$X_{A,x^*}=\overline{\{ (x_1^*t^{a_1},\dots, x_n^*t^{a_n}) \; | \; t \in (\CC^*)^d\}}.$$
The containment $X_{A,x^*}\cap \RR_{>0}^n \subseteq W\cap \RR_{>0}^n $ follows from Lemma \ref{lem:parameterization}. Now prove the other containment. For a positive real point $w\in W \cap \RR_{>0}^n$ we must have for each $k = 1,\dots, n-d$ that $F_k(w)=0$. Hence, by Definition \ref{def:PSSBin}, $$
w_{i_j}^{\lambda_{i_j}}w^{\lambda_- -\lambda_+} = \left(x^*_{i_j}\right)^{\lambda_{i_j}}\left(\prod_{l=1}^{j-1}\left(x^*_{i_l}\right)^{-\lambda_{i_l}}\right).
$$
For each circuit define the vector $\tilde{\lambda}\in\ZZ^n$ as the vector with non-zero entries $i_\ell$ corresponding to the values of the coefficients of $\lambda_{i_j}a_{i_j} -\sum_{y=1}^{j-1}\lambda_{i_y}a_{i_y}=0$.
and let $\Lambda$ be the matrix with rows $\tilde{\lambda}_k$; note that the rows of $\Lambda$ form a basis of $\ker(A)$. Set $\tilde{w}=(w_{i_1}, \dots, w_{i_j})$ and set $\tilde{x^*}=(x^*_{i_1}, \dots, x^*_{i_j})$; then $$
\left(\frac{\tilde{w}}{\tilde{x^*}}\right)^{\tilde{\lambda}_k}=1, \;\;\; {\rm which\; gives, \;\;}\tilde{\lambda}_k\cdot \log(\tilde{w}/\tilde{x^*})=0.
$$ Further, $\log(\tilde{w}/\tilde{x^*})\in \ker(\Lambda)$ and, hence, $\tilde{w}/\tilde{x^*}$ is in the image of $x^* t^A$.
\end{proof}


As a result of Theorem \ref{thm:Agreement} we can restrict the study of positive steady states to the study of the PSS matroid.
The next proposition establishes a connection between the PSS matroids of all members of a family.

\begin{proposition}
Fix a family of toric reaction networks $\mathcal{N}$ and let $ \mathcal{N}_M, \, \mathcal{N}_N\in\mathcal{N}$ with $M<N$. If both members of the family have the same number of conservation relations then $\mathcal{M}(A_M)$ is a submatroid of $\mathcal{M}(A_N)$.
\label{propn:ASubMat}
\end{proposition}
\begin{proof}
If the dimensions of the varieties are the same, then by Theorem \ref{thm:MatroidHom} the $A$-matrix of $\mathcal{N}_M$ is a submatrix of the $A$-matrix of $\mathcal{N}_N$. Therefore, $E_M \subset E_N$ and, trivially, any set of linearly independent columns of any matrix is also linearly independent set of columns in any submatrix containing these. Hence, $\mathcal{I}_M = \mathcal{I}_N\cap\mathcal{P}(E_M)$.
\end{proof}

\subsection{Experimental Design and Compatibility}\label{Sec:ExpermentialDesign}

We now study how steady state matroids and submatroids can be employed in experimental design and model rejection. For the remainder of this subsection we assume that the number of conservation relations within a family is fixed. Hence, by Proposition \ref{propn:ASubMat} the matroids of smaller family members are submatroids of the matroids of larger family members.

Previous related work includes the study of ``complex-linear steady state invariants'' \cite{Karp2012} and data coplanarity \cite{Harrington2012}. A study using the language of algebraic matroids explicitly can be found in \cite{Gross2016}. In this section we obtain results similar to \cite{Harrington2012,MacLean2015} using the PSS matroid and the binomials of Definition \ref{def:PSSBin} and show how they can be used in model rejection and experimental design for entire toric families. In particular, we focus on two extreme cases, the one of all parameters being known and the other of perfect, noise-free measurements. Statistical version of our results could form part of a future work.

First, we determine which species need to be measured to be able to construct the steady state locus for an entire family (if the rate constants are known). 
\begin{proposition}
Fix a family of chemical reaction networks $\mathcal{N}$ for which the rate constants $\kappa$ associated to every family member are known. There exists a subset of species of the smallest family member which is sufficient to measure at steady state in order to construct a corresponding positive steady state for every subsequent family member.
\label{propn:SmallestSets}
\end{proposition}
\begin{proof}
Recall that $x^*$ is a positive vector for known constants. Choose a basis $S$ of the PSS matroid of the smallest family member $\mathcal{N}_1$. Hence, by Definition \ref{def:PSSBin} binomial relations can be constructed to determine the steady state concentrations of the chemical species not in the basis. By Proposition \ref{propn:ASubMat} any basis of $\mathcal{N}_1$ is a basis of the subsequent family members and, hence, Definition \ref{def:PSSBin} applies.
\end{proof}
Hence, by Proposition \ref{propn:SmallestSets}, measuring the concentrations of all species in a basis of the smallest network in a family is sufficient to determine the expected values of the chemical concentrations at steady state of the entire family.

\begin{example}
A detailed study of positive steady state parameterizations for the processive system has been carried out in \cite{ProcPhos}. Suppose we know all $11$ rate constants of the one and two site processive phosphorylation networks of the running example. In this case the functions $x^*_i(\kappa)$ can be found and evaluated explicitly (e.g. see \cite{ProcPhos}). Further, the columns $a_1,a_2,a_3$ form a basis of $\RR^3$ and hence the column space of the $A$ matrix. Therefore, given measurements of $x_1,x_2,x_3$, the expected steady state concentrations of any species $x_j$ with $j>3$ can be found as a zero of the Laurent binomial
$$\Phi(C)(x) = x_1^{\lambda_1}x_2^{\lambda_2}x_3^{\lambda_3} - \left(x_1^*\right)^{\lambda_1}\left(x_2^*\right)^{\lambda_2}\left(x_3^*\right)^{\lambda_3}\left(x_j^*\right)^{-\lambda_j} x_j^{\lambda_j}.$$
\end{example}

Next, we use the PSS matroids for model selection or model rejection for perfect, that is noise-free, data. However, the result of Lemma \ref{lem:Rejection} could also be applied to noisy data by following the construction in \cite{MacLean2015}. To determine whether a model is compatible with observed data it is necessary to determine whether there exists a set of parameters $\{\kappa_1,\dots,\kappa_m\} > 0$ such that a measured data point $\{\xi_{i_1},\dots,\xi_{i_j}\}$ is an element of the steady state variety.
We proceed by formulating a parameter-free condition for model compatibility.

\begin{lemma}
Let $\mathfrak{N}$ be a reaction network with PSS matroid $\mathcal{M}(A)$. Assume that the rate constants are fixed but unknown, but the initial conditions are varied. Fix a circuit $C$ corresponding to the linear relations among the columns of $A$ via $\sum_{l=1}^{j-1}\lambda_{i_l}a_{i_l} = \lambda_{i_j}a_{i_j}$.
Given a pair of steady state measurements $\{\xi_{i_1},\dots,\xi_{i_j}\}$ and $\{\zeta_{i_1},\dots,\zeta_{i_j}\}$ of the concentrations of $C$, the corresponding model is compatible only if $$\xi_{i_j}^{\lambda_{i_j}}\prod_{l=1}^{j-1}\xi_{i_l}^{-\lambda_{i_l}} = \zeta_{i_j}^{\lambda_{i_j}}\prod_{l=1}^{j-1}\zeta_{i_l}^{-\lambda_{i_l}}.$$
\label{lem:Rejection}
\end{lemma}
\begin{proof}
As in Remark \ref{rem:Constants} we fix rate constants $\kappa = (\kappa_1,\dots,\kappa_m)^T\in \RR^m_{>0}$ such that $x^*\in \RR^n_{>0}$. Rearrange $\Phi(C)$ of Definition \ref{def:LaurentBin} to give
$$
x_{i_j}^{\lambda_{i_j}}\prod_{l=1}^{j-1}x_{i_l}^{-\lambda_{i_l}} = \left(x^*_{i_j}\right)^{\lambda_{i_j}}\left(\prod_{l=1}^{j-1} \left(x^*_{i_l}\right)^{-\lambda_{i_l}}\right)=\theta \in \RR_{>0}.
$$
For the measurements to be compatible we must have that when we evaluate the expression above at $x_{i_l}=\xi_{i_l}$ and at $x_{i_l}=\zeta_{i_l}$ we obtain the same value, $\theta$. The conclusion follows. 
\end{proof}


\begin{example}
Suppose we wanted to test whether some data we obtained could originate from a processive phosphorylation system. We set up the experiment using two different initial conditions and measured the concentrations of the elements of the circuit $\{x_1,x_2,x_3,x_4\}$ and based on the results of the previous section we know that $(x_1x_3)/(x_2x_4) =\, const.$ Let the two ideal data points be $\xi = \{5/2,3/10,10/8,875/96\}$ and $\zeta = \{1/2,3/5,1/8,35/384\}$. A priori these data seem unrelated, however, by using the circuit information we find that
$$
\frac{\xi_1\xi_3}{\xi_2\xi_4} = \frac{\zeta_1\zeta_3}{\zeta_2\zeta_4}.
$$
Therefore, the data are compatible with the processive model. Note that the relation holds equally for the one-site and two-site model and we can conclude that the entire family is compatible.
\end{example}

Due to the additional structure provided by the PSS matroid the condition of Lemma \ref{lem:Rejection} is much simpler than the linear algebra condition of \cite{Harrington2012}. By Proposition \ref{propn:ASubMat} it is easy to see that if Lemma \ref{lem:Rejection} holds for a given PSS matroid it holds for all of its submatroids. Hence, measuring only a subset of species which belongs to a small member of the family tells us that the measurements of this subset of species is compatible for all subsequent family members. This allows us to determine that a given data set is compatible with a family of networks, but we cannot specify which network in the family is `most' compatible with the data. 

Identifying model parameters for perfect data has been studied extensively in previous work \cite{Meshkat2018,Chis2011,Walter2014} and, therefore, the discussion in this chapter is restricted to the novel methods obtained by using the PSS matroid of a family of networks. In particular, the PSS matroid allows for straightforward conversion of a variety in the space of species to a variety in the space of parameters. We first show that, in the case of toric steady states the biochemically viable parameter sets, $\kappa = (\kappa_1,\dots,\kappa_m)^T\in\RR^m_{>0}$, are the positive part of an algebraic variety and then generalise this result to the entire family.

\begin{proposition}
Let $A$ be a full rank $d\times n$ integer matrix and let $C_1,\dots,C_{n-d}$ be a collection of circuits of the matroid $\mathcal{M}(A)$, each containing the same basis $S$. Using Definition \ref{def:LaurentBin} to obtain the ideal $J = \langle \overline{\Phi(C_1)},\dots,\overline{\Phi(C_{n-d})}\rangle \subseteq R=\RR(\kappa_1,\dots,\kappa_m)[x_1,\dots,x_n]$ and denoting the variables present in a circuit as $x(C_1)\subseteq \{x_1,\dots,x_n\}$, then the intersection ideal $J_{C_i} = J\cap\RR(\kappa_1,\dots,\kappa_m)[x(C_i)]$ is principal (generated by one element) with generator $\overline{\Phi(C_i)}$.
\label{lem:CircuitIdeal}
\end{proposition}
\begin{proof}
By construction both, the numerator and the denominator of $\Phi(C_i)$ contain only variables in $C_i$ and, hence, after clearing the denominator the resulting polynomial $\overline{\Phi(C_i)}$ also only contains variables in $C_i$.
Since $C_i$ is a circuit, the ideal $J_{C_i}$ has codimension 1 in $R\cap\RR(\kappa_1,\dots,\kappa_m)[x(C_i)]$ and, hence, by \cite[I.,\S 7, Proposition 4]{MumfordRed} it is principal. Further, $J\cap\RR(\kappa_1,\dots,\kappa_m)[x(C_i)] = \langle\overline{\Phi(C_i)}\rangle$ and, therefore, $J_{C_i} = \langle\overline{\Phi(C_i)}\rangle$.
\end{proof}

Proposition \ref{lem:CircuitIdeal} shows that the positive part of the steady state variety can easily be projected onto coordinate subspaces by dropping circuits. Hence, the PSS matroid allows for some freedom to ``pick and mix'' variables according to measurements. The picking and mixing corresponds to the geometric operation of projection of the variety $X = V(\langle \overline{\Phi(C_1)},\dots,\overline{\Phi(C_{n-d})}\rangle) \subseteq \RR(\kappa_1,\dots,\kappa_m)^n$. Next, suppose there exists a measurement $\xi = \{\xi_{i_1},\dots,\xi_{i_\ell+d}\}$ containing measurements for the species in a basis $S$ of the PSS and the species in the circuits $C_1,\dots,C_\ell$, all containing $S$. Combining the idea of projection and measurement (evaluation) leads to the definition of a {\em parameter variety}.

\begin{definition}
Keeping the same notation as above and, by  choosing an appropriate set of generators, i.e.~``clearing the denominators'', let $J = \langle \overline{\Phi(C_1)},\dots,\overline{\Phi(C_{\ell})}\rangle\subseteq \RR[\kappa_1,\dots,\kappa_m,x_1,\dots,x_n]$. Hence, $V(J) \subset \RR^m\times\RR^n$. The parameter variety, $X_m$, is obtained from $V(J)$ by the evaluation $$X_m = V(J)\cap V(\langle x_{i_1}-\xi_{i_1},\dots,x_{i_{\ell + d}}-\xi_{i_{\ell + d}}\rangle)\subseteq\RR^m.$$
\end{definition}

The parameter variety is obtained by the selective projection and evaluation of the binomials obtained from Definition \ref{def:PSSBin} and is a variety in the space of parameters only. Every parameter vector compatible with the measurement $\xi$ is on the parameter variety. Every set of measurements gives rise to a parameter variety and, in order to be compatible with a model, their intersection needs to be non empty, in fact, their intersection needs to be non empty in the positive orthant. In order to uniquely identify a model based on a given measurement the positive orthant of $X_m$ needs to consist of a single point only. Various algebraic techniques can be applied to show when this is the case, e.g.~\cite{CLO,Muller2016}.

The parameter varieties of families of toric networks can be related by applications of projections (selective application of Definition \ref{def:PSSBin}) and the partial evaluation map $\pi$. Suppose $J_N$ and $J_\text{int}$ are the ideals of the $N^\text{th}$ member of a family and the intermediate model between the $N^\text{th}$ and $(N+1)^\text{th}$ member, respectively. Let both ideals (or a projection of them) contain the same circuits $C_1,\dots,C_\ell$. Then, by Proposition \ref{propn:BinId}, $J_{N} = \pi(J_\text{int})$.

\section{Inheritance of Multistationarity for Toric Families}\label{sec:Multi}

In this section we investigate the inheritance of multistationarity among members of families of toric chemical reaction networks. 
Our main result is Theorem \ref{thm:MultiStat}, where we apply results of \cite{Sadeghimanesh2018,Muller2016} to show  that if we can find a multistationary member of a family (satisfying certain conditions), then every larger member of the same family is multistationary for some parameter values.
We begin by introducing some notation.

\begin{definition}
Let $I_B\subset \RR[x_1,\dots,x_n]$ be a prime binomial ideal defining a complete intersection of codimension $n-d$ and let the matrix $B$ be an exponent matrix of a minimal generating set of $I_B$. Define
$$J = \begin{pmatrix}Z\\ B^T\end{pmatrix}$$ where $Z$ is an integer matrix. Further, let $$J^\lambda = \begin{pmatrix}Z\\ \left(B^T\right)^\lambda \end{pmatrix},$$ where $\left(B^T\right)^\lambda = (b_1\lambda_1,\dots,b_n\lambda_n)$ for the columns $b_i$ of  $B^T$. We call $J^\lambda$ regular if $\det(J^\lambda) \neq 0$ for some values of $(\lambda_1,\dots,\lambda_n)$.
\label{def:J}
\end{definition}

We now state the main theorem of this section. 

\begin{theorem}
Fix a family of toric chemical reaction networks. Suppose that the family obeys the following conditions.
\begin{enumerate}
  \item[(C1)] The family has toric intermediates.
  \item[(C2)] The number of conservation relations stays constant.
  \item[(C3)] The matrix $J^\lambda$ for the $N^\text{th}$ network is regular.
\end{enumerate}
Then, if the $N^\text{th}$ member of the family is capable of multistationarity, every member of the family for which $M\geq N$ is also capable of multistationarity.
\label{thm:MultiStat}
\end{theorem}

Before proving the theorem above we prove the following lemma.

\begin{lemma}
Fix a family of reaction networks, where each member of the family and each intermediate network have the same number of conservation relations. Let $Z_N$ be a conservation relation matrix of the $N^\text{th}$ network. Then, $Z_{N+1}$ is obtained by adding columns to $Z_N$, i.e. $Z_N$ is a submatrix of $Z_{N+1}$.
\label{lem:ConsLaws}\end{lemma}
\begin{proof}
	First, build the stoichiometric matrix of the intermediate network by adding $n'$ species and $m'$ reactions to the $N^{\text{th}}$ network. Hence, the stoichiometric matrix of the intermediate network, $\Gamma^{\text{int}}_N$, has the form $$\Gamma_N^{\text{int}} = \left(
	\begin{array}{cc}
	\begin{array}{c|} \Gamma_N\\ \hline 0_{n'\times m} \end{array} & \tilde{\Gamma}_{N}
	\end{array}\right),$$
	where $\Gamma_N = (r_1,\dots,r_m)$ is the $n\times m$ stoichiometric matrix of the $N^{\text{th}}$ network and, $\tilde{\Gamma}_{N} = (r_{m+1},\dots,r_{m+m'})$ is the $(n+n')\times m'$ matrix of new reactions. Hence, any basis of the left kernel of $\Gamma^{\text{int}}_N$ can be written as $\tilde{z}_i = (z_i\:,\:z'_i)$, where $z_i$ is an element of a basis of the left kernel of $\Gamma_N$ or zero and $z'_i$ is a vector to be determined by $\tilde{\Gamma}_N$. The assumption of equidimensionality of the left kernels of $\Gamma_N$ and $\Gamma_N^{\text{int}}$ can be satisfied only if $\text{rank}(\tilde{\Gamma}_N) = n'$ and, in the same fashion as Corollary \ref{cor:ADel}, it follows, that $Z_N$ is a submatrix of $Z_N^{\text{int}}$. Now, delete $m''$ of the first $m$ columns of $\Gamma_N^{\text{int}}$ to form $\Gamma_{N+1}$. By the equidimensionality assumption it follows, that the rows of $Z_N^{\text{int}}$ form a basis of the left kernel of $\Gamma_{N+1}$. Hence, $Z_N^{\text{int}} = Z_{N+1}$ and $Z_N$ is a submatrix of $Z_{N+1}$.
\end{proof}

From Lemma \ref{lem:ConsLaws} and Theorem \ref{thm:Agreement} it becomes apparent why (C2) and toric steady states are required. Namely, toric steady states are needed to construct a matrix $B$ and, hence, $J_N^\lambda$ and (C2) is required to ensure that the conservation relation matrices are submatrices. The proof of Theorem \ref{thm:MultiStat} is now stated.

\begin{proof}
Fix a vector $x^*\in \RR_{>0}^n$ as in Definition \ref{def:LaurentBin} and choose a basis $S$ and $n-d$ circuits $\{C_1,\dots,C_{n-d}\}$ of the PSS matroid associated to the $N^\text{th}$ member of a family such that $\bigcap_i C_i = S$. Consider the polynomial system \begin{equation}
    \overline{\Phi(C_1)}(x)=\cdots=\overline{\Phi(C_{n-d})}(x)=Z_N\cdot x-c=0. \label{eq:PhiSystem}
\end{equation} 
Further, let $B_N$ denote the exponent matrix of the binomials $\overline{\Phi(C_1)},\dots,\overline{\Phi(C_{n-d})}$. Hence, following the construction of Definition \ref{def:J} we obtain
the square matrix $$J^\lambda_N = \begin{pmatrix}Z_N\\
\left(B_N^T\right)^\lambda\end{pmatrix}.$$
By construction (see Theorem \ref{thm:Agreement}), $V(\overline{\Phi(C_1)}(x),\cdots,\overline{\Phi(C_{n-d})}(x))\cap \RR_{>0}^n \neq \emptyset$. By \cite[Theorem 2.7]{Sadeghimanesh2018} the system \eqref{eq:PhiSystem} is multistationary if and only if either $\det(J^\lambda_N) = 0$ or $\det(J^\lambda_N) \neq 0$ and the polynomial $\det(J^\lambda_N)$ in $\lambda_1,\dots, \lambda_n$ has a positive and a negative term. Suppose the $N^{\text{th}}$ network is multistationary and, by condition (C3),  $\det(J^\lambda_N) \neq 0$. Next, build the $(N+1)^\text{th}$ network by adding $n'$ new species and consider its matrix $J_{N+1}^\lambda$. To obtain this matrix consider the new circuit binomials $\overline{\Phi(C_{n-d+1})}(x)=\cdots=\overline{\Phi(C_{n+n'-d})}(x)=0$, where each circuit contains the basis $S$ and one new variable $x_{n+i}$ with its corresponding positive exponent $b_{n+i}$, where $i = 1,\dots,n'$. The exponent vectors will therefore have exactly one non-zero element (namely the positive integer $b_{n+i}$) in the last $n'$ rows. Therefore, the matrix $J_{N+1}^\lambda$ has the block form
$$
J^\lambda_{N+1} =
\left(\begin{array}{ccc}
J_N^\lambda &\vline & \begin{matrix} Z_{N+1}|_{[(1\dots d)\times(n+1\dots n+n')]}\\ 0_{(n-d)\times l}\end{matrix}\\
\hline
\left(B_{N+1}^T\right)^\lambda|_{[(n-d+1\dots n+n'-d)\times (1\dots n)]} &\vline & {\rm diag}(b_{n+1}\lambda_{n+1},\dots,b_{n+n'}\lambda_{n+n'})
\end{array}\right)
$$
where $A|_{[y_1\dots y_m\times a_1\dots a_n]}$ denotes the restriction of a matrix $A$ to the rows $y_1\dots y_m$ and columns $a_1\dots a_n$. The matrix ${\rm diag}(b_{n+1}\lambda_{n+1},\dots,b_{n+n'}\lambda_{n+n'})$ is a diagonal matrix with diagonal elements $b_{n+1}\lambda_{n+1},\dots,b_{n+n'}\lambda_{n+n'}$. Hence, the expression $T = \left(\prod_{i=n+1}^{n+n'} b_i\lambda_i\right) \det(J^\lambda_N) \neq 0$ must appear in $\det(J^\lambda_{N+1})$ and, in particular, no term in $T$ can be cancelled by any other term appearing in $\det(J^\lambda_{N+1})$. Hence, if $\det(J^\lambda_N)$ has coefficients of opposite sign so does $\det(J^\lambda_{N+1})$ and, therefore the network is multistationary.
The proof is completed by induction.
\end{proof}

\begin{remark}
Note that, since we are allowed to choose different bases of the PSS matroid, and hence, different binomial systems it may be possible to satisfy condition (C3) of Theorem \ref{thm:MultiStat} for one particular choice of basis but not for a different choice of basis.
\end{remark}

We illustrate our results using the two and three site distributive phosphorylation networks.

\begin{example}
The PSS matroid of the one-site and two-site distributive phosphorylation networks are represented by the $A$-matrices
$$
A_1 =
\begin{pmatrix}
1 &0 &0 &1 &1 &1\\
1 &1 &0 &0 &1 &1\\
0 &0 &1 &1 &1 &1\\
\end{pmatrix}\text{ and }
A_2 =
\begin{pmatrix}
1 &0 &0 &1 &1 &1 &2 &2 &2\\
1 &1 &0 &0 &1 &1 &0 &1 &1\\
0 &0 &1 &1 &1 &1 &1 &1 &1\\
\end{pmatrix},
$$
respectively. Hence, choosing a basis of $a_1 = (1,1,0)^T$, $a_2 = (0,1,0)^T$ and $a_3 = (0,0,1)^T$ we find a parameterization for the one-site network as
\begin{align*}
    x_2x_4-x^*_2x^*_4\left(x^*_1x^*_3\right)^{-1}x_1x_3 &= 0,\\
    x_5 - \left(x^*_1x^*_3\right)^{-1} x^*_5 x_1x_3 &= 0,\\
    x_6 - \left(x^*_1x^*_3\right)^{-1}x^*_6 x_1x_3 &= 0.
\end{align*}
The two-site model has three additional equations, namely
\begin{align*}
    x_2^2x_7 - \left(x^*_1\right)^{-2}\left(x^*_3\right)^{-1}\left(x^*_2\right)^2 x^*_7 x_1^2x_3 &= 0,\\
    x_2x_8 - \left(x^*_1\right)^{-2} \left(x^*_3\right)^{-1}x^*_2x^*_8 x_1^2x_3 &= 0,\\
    x_2x_9 - \left(x^*_1\right)^{-2} \left(x^*_3\right)^{-1}x^*_2x^*_9 x_1^2x_3 &= 0.
\end{align*}
Hence, we get the $B$-matrices
\begin{equation*}
B_1^T =
\begin{pmatrix}
-1 & 1 & -1 & 1 & 0 & 0\\
-1 & 0 & -1 & 0 & 1 & 0\\
-1 & 0 & -1 & 0 & 0 & 1\\
\end{pmatrix}\text{ and }\;
B_2^T =
\begin{pmatrix}
-1 & 1 & -1 & 1 & 0 & 0 & 0 & 0 & 0\\
-1 & 0 & -1 & 0 & 1 & 0 & 0 & 0 & 0\\
-1 & 0 & -1 & 0 & 0 & 1 & 0 & 0 & 0\\
-2 & 2 & -1 & 0 & 0 & 0 & 1 & 0 & 0\\
-2 & 1 & -1 & 0 & 0 & 0 & 0 & 1 & 0\\
-2 & 1 & -1 & 0 & 0 & 0 & 0 & 0 & 1
\end{pmatrix}.
\end{equation*}
The conservation relations can be chosen as
\begin{equation*}
    Z_1 =
    \begin{pmatrix}
    -1 & -1 & 1 & 1 & 0 & 0\\
    1 & 0 & 0 & 0 & 1 & 0\\
    0 & 1 & 0 & 0 & 0 & 1
    \end{pmatrix}
    \text{ and }\;
    Z_2 =
    \begin{pmatrix}
     -1 & -1 & 1 & 1 & 0 & 0 & 1 & 0 & 0\\
    1 & 0 & 0 & 0 & 1 & 0 & 0 & 1 & 0\\
    0 & 1 & 0 & 0 & 0 & 1 & 0 & 0 & 1
    \end{pmatrix}.
\end{equation*}
This leads to
$$\text{det }(J^\lambda_1) = \lambda_3 \lambda_4 \lambda_5+\lambda_1 \lambda_4 \lambda_6+\lambda_3 \lambda_4 \lambda_6+\lambda_3 \lambda_5 \lambda_6+\lambda_4 \lambda_5 \lambda_6.$$
Note that this determinant is square-free and homogeneous as investigated in \cite{Dickenstein2018} and that it has only coefficients equal to $+1$. Indeed, it is a well known fact the the one-site distributive network is monostationary \cite{Flockerzi2014}. The determinant of $J^\lambda_2$ is equal to
\begin{align*}
\text{det}(J^\lambda_2) = &-\lambda_2 \lambda_3 \lambda_4 \lambda_5 \lambda_7 \lambda_8-\lambda_1 \lambda_2 \lambda_3 \lambda_6 \lambda_7 \lambda_8-\lambda_1 \lambda_2 \lambda_4 \lambda_6 \lambda_7 \lambda_8-\lambda_1 \lambda_3 \lambda_4 \lambda_6 \lambda_7 \lambda_8-\lambda_2 \lambda_3 \lambda_4 \lambda_6 \lambda_7 \lambda_8\\
&-\lambda_2 \lambda_4 \lambda_5 \lambda_6 \lambda_7 \lambda_8+\lambda_3 \lambda_4 \lambda_5 \lambda_6 \lambda_7 \lambda_8-\lambda_1 \lambda_2 \lambda_3 \lambda_4 \lambda_5 \lambda_9-\lambda_1 \lambda_2 \lambda_3 \lambda_5 \lambda_6 \lambda_9-2 \lambda_1 \lambda_2 \lambda_4 \lambda_5 \lambda_6 \lambda_9\\
&-\lambda_2 \lambda_3 \lambda_4 \lambda_5 \lambda_6 \lambda_9+\lambda_1 \lambda_3 \lambda_4 \lambda_5 \lambda_7 \lambda_9+\lambda_1 \lambda_3 \lambda_5 \lambda_6 \lambda_7 \lambda_9+2 \lambda_1 \lambda_4 \lambda_5 \lambda_6 \lambda_7 \lambda_9+\lambda_3 \lambda_4 \lambda_5 \lambda_6 \lambda_7 \lambda_9\\
&-2 \lambda_2 \lambda_3 \lambda_4 \lambda_5 \lambda_8 \lambda_9-\lambda_1 \lambda_2\lambda_3 \lambda_6 \lambda_8 \lambda_9-2 \lambda_1 \lambda_2 \lambda_4 \lambda_6 \lambda_8 \lambda_9-\lambda_1 \lambda_3 \lambda_4 \lambda_6 \lambda_8 \lambda_9-2 \lambda_2 \lambda_3 \lambda_4 \lambda_6 \lambda_8 \lambda_9\\
&-\lambda_2 \lambda_3 \lambda_5 \lambda_6 \lambda_8 \lambda_9-2 \lambda_2 \lambda_4 \lambda_5 \lambda_6 \lambda_8 \lambda_9+ \lambda_3 \lambda_4 \lambda_5 \lambda_6 \lambda_8 \lambda_9+ \lambda_7 \lambda_8 \lambda_9 \text{ det}(J^\lambda_1),
\end{align*}
and, therefore, contains a term $T = \lambda_7\lambda_8\lambda_9\text{det}(J^\lambda_1)$. The determinant of the two-site network has coefficients of opposite signs and, hence, the network is multistationary and so are all larger networks in the family. In particular, it is known that the $N$-site distributive network has a maximum of $2N-1$ positive steady states \cite{Flockerzi2014}.
\label{ex:DistPhosMulti}
\end{example}


\section{Conclusion}\label{sec:Conclusion}

In this paper we studied families of chemical reaction networks with toric steady states, which we called toric families. First, we investigated the dimensions and parameterizations of toric steady state varieties and connected them to network properties whenever possible. In particular, if there is a finite number of positive steady states for all choices of $c_1,\dots,c_d$, then the number of conservation relations determines the dimension of the toric steady state variety and, with certain restrictions, the monomial parameterization of a chemical species $X_i$ is preserved throughout the family.

We next studied the PSS matroid defined by the parameterization of the positive steady states. In particular, we proved that the algebraic matroid defined by the binomial steady state ideal is equivalent to the PSS matroid. We showed how the PSS matroid can be decorated with binomials and how they can be used for model selection, experimental design or even parameter identification. 

The final section investigated the multistationarity structure of toric families. The main result of the section showed that, under some mild restrictions, if a member of a family is capable of multistationarity then all larger members are too.
This result was proved using the binomials constructed from the PSS matroids. 
We illustrated our results on the multisite distributive phosphorylation network.

Further research could include applying the results of this paper to other meaningful biochemical families such as different models for immune system reactions, e.g. \cite{McKeithan1995,Dushek2014}. Another direction could be to study the parameter varieties defined in this paper in the context of previous identifiability research and aim to include noisy data. 

\section*{Acknowledgements}

The authors would like to thank the American Institute of Mathematics (AIM) for initiating this collaboration, the Erwin Sch{\"o}dinger Institute, Vienna (ESI) and the Institute for Computational and Experimental Research in Mathematics (ICERM) where part of the research was carried out. MFA would like to thank the EPSRC for supporting this research through grant EP/G03706X/1 and Merton College Oxford.
MH was partially supported by the Independent Research Fund of Denmark. 
The authors would like to than Elisenda Feliu and Heather Harrington for helpful comments on the manuscript.

\section{Appendix A: Additional Proofs}

\subsection{Proof of Lemma \ref{lem:SatEval}}
\begin{proof}
	Note that $\overline{V(L)-V(x_1\cdots x_{n+n'})}=V(L)$.
	Hence,
	\begin{align*}
	 \overline{X_{N+1}\cap V(L) - V(x_1\cdots x_{n+n'})}   &=\overline{(X_{N+1}-V(x_1\cdots x_{n+n'}))\cap (V(L)-V(x_1\cdots x_{n+n'}))}\\
	 &=\overline{(X_{N+1}-V(x_1\cdots x_{n+n'}))}\cap\overline{(V(L)-V(x_1\cdots x_{n+n'}))}\\
	 &=\overline{(X_{N+1}-V(x_1\cdots x_{n+n'}))}\cap V(L),
	\end{align*}
	which proves the first statement. For the second statement note that, by construction of the evaluation map, $X_{N+1}\cap V(L) = X_N \subset \CC^m_\kappa\times\CC^n_x\times\CC^{n'}$. Hence,
	\begin{align*}
	\overline{X_{N+1}-V(x_1\cdots x_{n+n'})}\cap V(L) &= \overline{X_{N+1}\cap V(L) - V(x_1\cdots x_{n+n'})}\\
	&= \overline{X_N - V(x_1\cdots x_{n+n'})}\\
	&= \overline{X_N - V(x_1\cdots x_n)}.
	\end{align*}
\end{proof}

\subsection{Proof of Proposition \ref{propn:BinId}}
\begin{proof}
By assumption $I_{N+1}:(x_1\cdots x_{n+n'})^\infty$ is a binomial ideal, hence, $$
I_{N+1}:(x_1\cdots x_{n+n'})^\infty=(\Xi_{1}^+x^{b_1^+}-\Xi_{1}^-x^{b_1^-}, \dots, \Xi_{\nu}^+x^{b_\nu^+}-\Xi_{\nu}^-x^{b_\nu^-}),
$$ where $\Xi_j^\pm$ are polynomial functions of the rate constants. 
For each term $\beta_j=\Xi_{j}^+x^{b_j^+}-\Xi_{j}^-x^{b_j^-}$ it holds that $\pi(\beta_j)=\pi(\Xi_{j}^+)x^{b_j^+}-\pi(\Xi_{j}^-)x^{b_j^-}$, hence either:\begin{itemize}
\item $\pi(\beta_j)=0$, this happens if $\pi(\Xi_{j}^+)=\pi(\Xi_{j}^-)=0$;
\item $\pi(\beta_j)$ is a monomial, this happens if one of $\pi(\Xi_{j}^+)$ or $\pi(\Xi_{j}^-)$ evaluates to zero and the other does not;
\item $\pi(\beta_j) = \beta'_j$ and still binomial, this happens when both $\pi(\Xi^+_j)\neq 0$ and $\pi(\Xi^-_j)\neq 0$.
\end{itemize}
If $\pi(\beta_j)$ is a monomial for any $j$ then the ideal $\pi(I_{N+1}):(x_1\cdots x_{n+n'})^\infty=(1)$, otherwise, it must be a binomial ideal generated by the $\beta_j$ such that $\pi(\beta_j)=\beta'_j$. By Lemma \ref{lem:SatEval} this ideal is the binomial ideal $I_N:(x_1\cdots x_{n})^\infty$.

Since for a fixed choice of rate constants $I_N:(x_1\cdots x_{n})^\infty$ is a binomial ideal in $\RR[x_1,\dots, x_n]$ and it is assumed that $V(I_N)$ has a nonempty intersection with the positive orthant $\RR_{>0}^n$ then the last statement follows by \cite[Proposition 5.2]{conradi2018multistationarity}. This Proposition states that if $I$ is a binomial ideal with variety $V(I)$ at most one of its irreducible components intersects $\RR^n_{>0}$.
\end{proof}

\subsection{Proof of Proposition \ref{thm:Submatrix}}
\begin{proof}
First statement (i) is proved. By assumption $I_{N+1}:(x_1\cdots x_{n+n'})^\infty$ is a binomial ideal, hence, $$
I_{N+1}:(x_1\cdots x_{n+n'})^\infty=(\Xi_{1}^+x^{b_1^+}-\Xi_{1}^-x^{b_1^-}, \dots, \Xi_{\nu}^+x^{b_\nu^+}-\Xi_{\nu}^-x^{b_\nu^-}).
$$
By Proposition \ref{propn:BinId} it follows that if $I_N:(x_1\cdots x_n)^\infty$ is not empty, then it is also a binomial ideal and its set of generators must appear in $\pi(\Xi_{1}^+x^{b_1^+}-\Xi_{1}^-x^{b_1^-}, \dots, \Xi_{\nu}^+x^{b_\nu^+}-\Xi_{\nu}^-x^{b_\nu^-})$. Hence, its matrix of exponents can be obtained by choosing  a subset (of size $\mu \leq \nu$) of the exponent vector pairs $({b_1^+}, {b_1^-}),\dots,({b_\nu^+}, {b_\nu^-}) $. By definition of the matrix $B$, it follows that $B_N$ is a submatrix of $B_{N+1}$, proving (i).

Note that no variables in $x_{n+1},\dots,x_{n+n'}$ appear in $I_N$ and, hence, the subset of exponent vectors $b_i = b_i^+-b_i^-$ defining $I_N:(x_1\cdots x_{n})^\infty$ have the general form $b_i = (b_{i_1},\dots,b_{i_n},0,\dots,0)^T$. 
Further, the entries of the remaining $\nu-\mu$ exponent vectors have at least one non-zero entry in the last $n'$ rows and these vectors can be collected into a matrix $\tilde{B}$.
Hence, $B_{N+1}$ may be written in block form as
$$
B_{N+1} = \left(
\begin{array}{cc}
\begin{array}{c|} B_N\\ \hline 0_{n'\times \mu} \end{array} & \tilde{B}
\end{array}\right).
$$
It is clear that the first $n$ entries of any basis vector of the left kernel of $B_{N+1}$ must form a vector in the left kernel of $B_N$. Hence, the left kernel of $B_{N+1}$ has a basis of the form $\tilde{a}_i = (a_i,\bar{a}_i)$, where $a_i$ is an element of a given basis of the left kernel of $B_N$ or the zero vector and $\bar{a}_i\in\ZZ^{n'}$ is a vector of variables to be determined. From the rank condition it is clear that the matrix $\tilde{B}$ has a rank of at most $n'$ and, for any given $i$, a system of linear equations is obtained, $\tilde{a}_i\tilde{b}=0$, for all columns $\tilde{b}$ of $\tilde{B}$. Due to the rank condition, this system has always at least one (rational) solution and, hence,  the collection of all solutions form a basis of the left kernel of $B_{N+1}$  to give the matrix $A_{N+1}$ (after clearing the denominators by multiplications of rows with scalars). Therefore, $A_N$ can be chosen as a submatrix of $A_{N+1}$, proving (ii).
\end{proof}

\subsection{Proof of Theorem \ref{thm:AFamSub}}
\begin{proof}
Part (i) follows immediately from Theorem \ref{thm:Submatrix}.
To prove (ii), suppose $I_N$ and $I_{N+1}$ are the (binomial) steady state ideals of two members of a family and $I'_{N+1}$ is the steady state ideal of their intermediate network. It follows from Theorem \ref{thm:Submatrix} that the $A$-matrices of $I_N$ and $I_{N+1}$ are submatrices of the $A$-matrix of the intermediate. Each column of the $A$-matrix corresponds to a variable $x_i$. Since both $I'_{N+1}$ and $I_{N+1}$ are ideals with species $\{x_1,\dots, x_{n+n'}\}$ then their associated $A$-matrices have the same number of columns. Further, by assumption, the number of conservation relations stays constant. Therefore, by Proposition \ref{prop:DimAndConserveRelations} the $A$ matrices associated to $I'_{N+1}$ and $I_{N+1}$ are the same, since the ideals have the same dimension, meaning the $A$-matrices must have the same number of rows. Hence, the $A$-matrices associated to $I_N$ and $I_{N+1}$ are submatrices; i.e. $A_N$ is a submatrix of $A_{N+1}$.
\end{proof}

\subsection{Proof of Theorem \ref{thm:MatroidHom}}
\begin{proof}
The variables of the ground set of $\mathcal{M}(A)$ have algebraic dependencies as defined in Proposition \ref{propn:PSSmatroid}. Hence, their algebraic dependencies can be found by solving the implicitization problem $I = J \cap \RR(\kappa_1,\dots,\kappa_m)[x_1,\dots,x_n]$ where $J=\langle x_1-x^*_1t^{a_1},\dots,x_n-x^*_nt^{a_n}\rangle$. However, the same ideal defines the implicit equations of a toric variety defined by $\psi_A: t\to (x^*_1t^{a_1},\dots,x^*_nt^{a_n})$. Hence, the algebraic relations between the monomials $\phi_A(x)$ are identical to the algebraic relations defined by the binomial ideal $I_b$. This implies that $\mathcal{M}(A) = \mathcal{M}(I_b)$.
\end{proof}

\bigskip

\begin{small}

\bibliographystyle{plain}

\end{small}
\end{document}